\documentclass[journal,10pt]{IEEEtran}

\IEEEoverridecommandlockouts                              
\usepackage[utf8]{inputenc}
\usepackage[T1]{fontenc}

\let\labelindent\relax
\usepackage[english]{babel}
\usepackage{xcolor}
\usepackage{graphicx}
\usepackage{url}
\usepackage[intlimits]{amsmath}
\usepackage{amsthm,thmtools}
\usepackage{amssymb}
\usepackage{mathrsfs}
\usepackage{array}
\usepackage[justification=centering]{caption}
\usepackage{enumitem}

\usepackage{interval}
\usepackage{psfrag}
\usepackage{balance}
\newtheorem{rem}{Remark}
\newtheorem{prop}{Proposition}
\newtheorem{theorem}{Theorem}
\newtheorem{lem}{Lemma}
\newtheorem{coro}{Corollary}
\newtheorem{definition}{Definition}

\DeclareMathOperator*{\argmin}{arg\,min}

\DeclareMathOperator*{\trace}{Tr}

\DeclareMathOperator*{\diag}{diag}
\DeclareMathOperator*{\rank}{rank}
\DeclareMathOperator*{\sat}{sat}

\renewcommand{\Re}{\mathbb{R}}

\renewcommand{\paragraph}[1]{\smallskip\noindent\textbf{#1.} }

\newcommand{\BM}{\begin{bmatrix}}
\newcommand{\EM}{\end{bmatrix}}

\newcommand{\BBM}{\big[\begin{matrix}}
\newcommand{\EEM}{\end{matrix}\big]}

\newcommand{\bbm}{[\begin{matrix}}
\newcommand{\eem}{\end{matrix}]}

\usepackage{setspace}

\title{An optimization framework for resilient batch estimation in Cyber-Physical Systems}

\author{Alexandre Kircher$^1$, Laurent Bako$^1$, Eric Blanco$^1$, Mohamed Benallouch$^2$ 
\thanks{$^1$ A. Kircher, L. Bako and E. Blanco are  with Universit\'{e} de Lyon, Laboratoire Amp\`{e}re (Ecole Centrale Lyon, CNRS UMR 5005), F-69134. 
        {\footnotesize\tt E-mails: alexandre.kircher, laurent.bako, eric.blanco@ec-lyon.fr}}%
\thanks{$^2$ M. Benallouch  is with Universit\'{e} de Lyon, ECAM Lyon, Lab ECAM, F-69321 Lyon, France.
        {\footnotesize \tt E-mail: mohamed.benallouch@ecam.fr}}%
}

\begin{document}
%

\maketitle

\begin{abstract}
This paper proposes a class of resilient state estimators for LTV discrete-time systems.  The dynamic equation of the system is assumed to be affected by a bounded process noise. As to the available measurements, they are potentially corrupted by a noise of both dense and impulsive natures. The latter in addition to being arbitrary in its form, need not be strictly bounded.   In this setting, we construct the estimator as the set-valued map which associates to the measurements, the minimizing set of some appropriate  performance functions. We consider a family of such performance functions each of which yielding a specific instance of the general estimator. It is then shown that the proposed class of estimators enjoys the property of resilience, that is, it induces an estimation error which, under certain conditions, is independent of the extreme values of the (impulsive)  measurement noise. Hence, the estimation error may be bounded while the measurement noise is virtually unbounded. Moreover, we provide several error bounds (in different configurations) whose expressions depend explicitly on the degree of observability of the system being observed and on the considered performance function. Finally, a few simulation results are provided to illustrate the resilience property.
\end{abstract}


\section{Introduction}
\paragraph{Context}
We consider in this work the problem of designing state estimators which would be resilient against an (unknown) sparse noise sequence affecting the measurements. By sparse noise we refer here to a  signal sequence which is of impulsive nature, that is, a sequence which is most of the time equal to zero,  except at a few instants where it can take on arbitrarily large values. 
The problem is relevant for example, in the supervision of  Cyber-Physical Systems (CPS)~\cite{cardenas_secure_2008}. In this application, the supervisory data may be collected by spatially distributed sensors and then sent to a distant processing unit through some communication network. During the transmission, the data may incur intermittent packet losses or adversarial attacks consisting in e.g., the injection of arbitrary signals. Beyond CPS, there are many other applications where the sparse noise model of uncertainty is relevant: robust statistics \cite{Huber-Book-09}, hybrid system identification \cite{Bako11-Automatica}, intermittent sensor fault detection, etc.

\paragraph{Related works} The problem of estimating the state of CPS under attacks has been investigated recently through many different approaches. 
 Since the measurements are assumed to be affected by a sequence of outliers which is sparse in time, a natural scheme of solution to the state estimation problem may be to first process the data so as to detect the  occurrences of the nonzero instances of that sparse noise, remove the corrupted data and then proceed with classical estimation methods such as the Kalman filter or the Luenberger type of observer \cite{mishra_secure_2017,pasqualetti_attack_2013}. While this approach  sounds a priori reasonable,  the main challenge remains to achieve an efficient detection and isolation of the outliers. Regarding the scenarios where the sporadic noise is modeled in a  probabilistic setting, there exists a body of interesting  results providing performance limits of estimation schemes \cite{sinopoli_kalman_2004,Mo15-TAC,Ren19}. \\  
Another category of approaches, which are inspired by some recent results in compressive sampling \cite{Candes08-SPM,Foucart13-Book}, rely on sparsity-inducing optimization techniques. A striking feature of these methods  is that they do not treat separately the tasks of detection, data cleaning and estimation. Instead, an implicit discrimination of the wrong data is induced by some  specific properties of the to-be-minimized cost function.  One of the first works that puts forward this approach for the resilient state estimation problem is the one reported in \cite{fawzi_secure_2014}. There, it is assumed that only a fixed number of sensors are subject to attacks (sparse over time but otherwise arbitrary disturbances). The challenge then resides in the fact that at each time instant, one does not know which sensor is compromised. Note however that the assumptions in \cite{fawzi_secure_2014} were quite restrictive as no \textit{dense} process noise or measurement noise (other than the \textit{sparse} attack signal) was considered. 
These limitations open ways for later extensions in many directions. For example, \cite{shoukry_event-triggered_2016} suggests  a reformulation which  is argued to reduce computational cost  by using the concept of event-triggered update  
; \cite{pajic_attack-resilient_2017} considers an observation model which includes dense noise along with the sparse attack signal. In \cite{chang_secure_2018},  the assumption of a fixed number of attacked sensors is relaxed. Finally,  the recent paper \cite{Han19-TAC} proposes a unified framework for analyzing resilience capabilities of most of these (convex) optimization-based estimators. Although a bound on the estimation error was derived in this paper, it is not quantitatively related to the properties (e.g., observability) of the dynamic system being observed. The state estimation problem treated there is rather viewed  as a linear regression problem similarly to \cite{bako_class_2017,Candes06-IT}.

\paragraph{Contributions}
The contributions of the current paper consist in the design and the analysis of a class of  optimization-based resilient estimators for Linear Time-Varying (LTV) discrete-time systems.  The available model of the system assumes bounded noise in both the dynamics and the observation equation with the latter being possibly affected, additionally, by an unknown but sparse attack signal. Contrary to the settings considered in some existing works, we  did not impose here any  restriction  on the number of sensors which are subject to attacks, that is, any sensor can be compromised at any time. Note also that no statistical significance is attached to the uncertainties modeled by noise. In this setting, by generalizing our previous work reported in \cite{kircher_analysis_2020}, the current paper proposes a general (robust) estimation framework for the state of LTV systems. We propose a class of state estimators  constructed  as the set-valued maps which associate to the output measurements, the minimizing set of some  appropriate  performance functions. A variety of performance functions are considered for the design of the estimator and handled in a unified analysis framework: convex nonsmooth/smooth loss functions and nonconvex saturated ones. 
Our main theoretical results concern the resilience analysis of the proposed class of estimators. 
 We show that the estimation error associated with the new class of estimators can be made, under certain conditions, insensitive to the (possibly very large) amplitude of the sparse attack signal. The proposed analysis relies on new quantitative characterizations of the observability property of the system whose state is being observed.   Although the derived error bounds may be conservative, they have the important advantage of being explicitly expressible in function of the properties of the considered dynamic system and those of the optimized loss functions. This makes them valuable qualitative tools for assessing the impact of the estimator's design parameters and that of the system matrices on the quality of the estimation.  
For example, the proposed error bounds reflect the intuition that the more observable the system is with respect to the new criteria, the larger the number of instances of gross values (of the output noise) it can handle and the smaller the values of the bounds. 
Finally the paper shows that for some choice of the design functions (loss functions), some instances of the proposed family of estimators enjoy the exact recoverability property in the particular situation where  the measurements are corrupted  only by sparse noise. We present a  condition for this property that can be numerically verified by means of convex optimization.    
Overall, in comparison with \cite{Han19-TAC} which also considers resilient estimation though in a linear regression setting, we (i) introduce here an alternative definition of resilience, (ii) characterize quantitatively the impact of intrinsic properties (observability) of the system being observed on the quality of the estimation (iii) derive an explicit expression of a bound on the estimation error. 

\paragraph{Outline} The rest of the paper is  structured as follows. We start by introducing in Section \ref{part:ii}, the settings for the resilient state estimation problem. We then define in Section~\ref{part:iii} the new class of optimization-based estimators proposed here to address this problem. The analysis of this new framework is presented in Section~\ref{part:iv}. In Section \ref{sec:Exact-Recoverability}, we further discuss the properties of a special constrained version of the initial class of estimators. 
In Section \ref{sec:Numerical-Eval}, we comment on the numerical verification of the conditions derived in the analysis part.  Some numerical results are described in Section ~\ref{part:v} and  finally, concluding remarks are given in Section~\ref{part:vi}.

\paragraph{Notation} $\Re_{\geq 0}$ (respectively $\Re_{>0}$) is the set of nonnegative (respectively positive) reals. $\Re_*$ designates the set of real numbers excluding zero. We note $\Re^a$ the set of (column) vectors with $a$ real elements and  $\Re^{a\times b}$, the set of real matrices with $a$ rows and $b$ columns. If $M\in \Re^{a\times b}$, then $M^{\top}$ will designate the transposed matrix of $M$. $I$ will refer to the (square)  identity matrix of appropriate dimension. The notation $\lVert \cdot \rVert$ will denote a norm over a given set (which will  be specified when necessary). $\lVert \cdot \rVert_p$  denotes the $\ell_p$ norm (for $p\geq 1$) or the $\ell_p$ quasi-norm (for $0<p<1$) defined for $z=\begin{pmatrix}z_1 & \cdots & z_a\end{pmatrix}$ in $\Re^a$ by $\left\|z\right\|_p=\left(|z_1|^p+\cdots +|z_a|^p\right)^{1/p}$. The limit of this when $p\downarrow 0$ gives the so-called $\ell_0$-norm $\left\|\cdot\right\|_0$ of $z$, i.e.,  the number of nonzero entries in $z$. Its limit when $p\uparrow +\infty$ gives the infinity norm denoted $\left\|z\right\|_\infty$ and returning the maximum value of the $\left|z_i\right|$. For $x\in \Re$, $e^x$ refers to the exponential function applied to $x$.\\
 If $\mathcal{S}$ is a set, then $\mathcal{P}(\mathcal{S})$ is the power set of $\mathcal{S}$. If $\mathcal{S}$ is a finite set, the notation $|\mathcal{S}|$ refers to the cardinality of $\mathcal{S}$. \\
\textit{$\mathcal{K}_{\infty}$ functions \cite{Kellett14}.} We name $\mathcal{K}_{\infty}$ the set of functions $f:\Re_{\geq 0}\rightarrow\Re_{\geq 0}$ which are continuous, zero at zero, strictly increasing and satisfy $\lim_{\lambda\rightarrow + \infty} f(\lambda)=+\infty$. If $f\in \mathcal{K}_{\infty}$, then it admits an inverse, denoted here $f^{-1}$, which also lies in $\mathcal{K}_{\infty}$.  Similarly, we use the notation $\mathcal{K}_{\sat,a}$  to denote the set of \textit{saturated functions} $f:\Re_{\geq 0}\rightarrow\interval{0}{a}$ which are continuous, zero at zero, strictly increasing on $\interval{0}{a}$ and such that $f(\lambda)=f(a)$ for all $\lambda\geq a$.
%
\noindent \textit{Supremum.} Given a function $f$ over $\Re^a$ and a subset $\mathcal{S}$ of $\Re^a$, the notation $\sup_{z\in\mathcal{S}} f(z)<b$, with $b\in\Re$, will mean that for all $z$ in $\mathcal{S}$, $f(z)<b$. This notation includes the case where the supremum is $b$ but is not attained by any element of $\mathcal{S}$. 

\section{The Resilient Estimation Problem}\label{part:ii}
Consider a discrete-time  Linear Time-Varying (LTV) system described by 
\begin{equation}\label{eq:sys}
\Sigma : \left\{\begin{array}{r  l}
x_{t+1} &= A_tx_t+w_t \\
y_t &=C_tx_t+f_t  
\end{array}
\right.
\end{equation}
where $x_t\in\Re^n$ is the state vector of the system at time $t$ and $y_t\in\Re^{n_y}$ is the output vector at time $t$; $\{A_t\}$ and $\{C_t\}$ are families of matrices with appropriate dimensions; $\{w_t\}$ is an \textit{unobserved bounded noise sequence}.
As to $\{f_t\}$, it is regarded here as an (unobserved) \textit{arbitrary noise sequence} affecting the measurements. For clarity of the exposition, it may be convenient  to view $f_t$ as a combination of two types of sequences: a \textit{bounded sequence} $\left\{v_t\right\}$ and a \textit{sparse sequence} $\left\{s_t\right\}$ (this decomposition is indeed always possible for an arbitrary noise signal).   
Hence, we may  write
\begin{equation}\label{eq:f=s+v}
	f_t=v_t+s_t, 
\end{equation}
where $v_t$ is a sensor noise of moderate amplitude and  $s_t$ is a sparse noise sequence in the sense that its (entrywise and/or timewise) components are mostly equal to zero but its  nonzero elements can take on (possibly) arbitrarily large values. Such a sparse sequence $\left\{s_t\right\}$ may account for adversarial attacks in the same spirit as in \cite{fawzi_secure_2014,Han19-TAC}, intermittent sensor faults, or data losses, in particular when a communication network is involved in the data acquisition-transmission chain.  In the sequel, we may  also  refer to $\left\{w_t\right\}$  and $\left\{v_t\right\}$ in \eqref{eq:sys} and \eqref{eq:f=s+v} as  dense noises and to the largest elements of $\left\{s_t\right\}$ as outliers.   

For the sake of simplicity, the sparse (and potentially arbitrary large) noise is assumed here to affect only  the measurement equation. Note however that the proposed analysis method can be extended to the more general scenario where the sparse noises  may affect both the dynamics and the measurements.

\paragraph{Problem}The problem considered in this paper is the one of estimating the states $x_0,\ldots,x_{T-1}$ of the system \eqref{eq:sys} on a time period $\mathbb{T}=\left\{0,\ldots,T-1\right\}$ given $T$ measurements $y_0,...,y_{T-1}$ of the system output. 
We shall seek  an accurate estimate of the state despite the uncertainties in the system equations \eqref{eq:sys} modeled by $w_t$ and $f_t$ the characteristics of which are informally described above. In particular, we would like the  to-be-designed  estimator to produce an estimate such that the estimation error is, when possible, independent of the maximum amplitude of $\left\{f_t\right\}$. Such an estimator will  be called \textit{resilient}, see Definition \ref{def:resilience} for a formal characterization of this property.

\noindent Denote with 
\begin{equation}
X_{T-1}=\begin{pmatrix}
x_0 &x_1 &\dots & x_{T-1}
\end{pmatrix}
\end{equation}
the actual state trajectory of the system $\Sigma$  on a finite time horizon of length $T$. 
Similarly, we use the notation
\begin{equation}
Y_{T-1}=\begin{pmatrix}
y_0 & y_1 &\cdots & y_{T-1}
\end{pmatrix}
\end{equation}
to refer to the collection of measurements on the same time horizon. 
The state estimation problem is approached here from an offline perspective, therefore $T$ is fixed. For the sake of simplicity, the $T$ index will be dropped from the variable names and it will be assumed that signal matrices without an index represent values on the period $\mathbb{T}=\{0,\ldots,T-1\}$. To simplify further the formulas, we also pose $\mathbb{T}'=\{0,\ldots,T-2\}$ while $\mathbb{S}=\left\{1,\ldots,n_y\right\}$ will be a set indexing the sensors (or the rows of the matrices $C_t$ in \eqref{eq:sys}). 
\section{Optimization-based approach to Resilient state estimation}\label{part:iii}

\subsection{The state estimator}

In this section we present an optimization-based framework for solving the state estimation problem defined above. 
To define formally the proposed state estimator, let us first introduce the to-be-minimized objective function.   Given the  matrices $\left\{(A_t,C_t)\right\}$ of the system \eqref{eq:sys} and  $T$ output measurements $Y=\begin{pmatrix}y_0 & \cdots & y_{T-1}\end{pmatrix}$, we consider a performance function 
$V_{\Sigma}:\Re^{n_y\times T} \times \Re^{n\times T} \rightarrow \Re_{\geq 0}$ defined by 
\begin{equation}\label{eq:V}
V_{\Sigma}(Y,Z)=\lambda\sum_{t\in\mathbb{T}'}\phi_t(z_{t+1}-A_tz_t)+ \sum_{t\in\mathbb{T}}\psi_t(y_t-C_tz_t)
\end{equation}
where $Z=\begin{pmatrix}z_0 & \cdots & z_{T-1}\end{pmatrix}\in\Re^{n\times T}$ is a hypothetical trajectory matrix with $z_i$ denoting the $i$-th column of $Z$; $\lambda>0$ is a user-defined parameter which aims at balancing the contributions of the two terms involved in  the expression of the performance index $V_{\Sigma}(Y,Z)$. $\{\phi_t\}$ and $\{\psi_t\}$ are two families of positive functions (called here loss functions) defined on $\Re^n$ and $\Re^{n_y}$ respectively. For the sake of simplicity, we will assume throughout the paper that for all $t$ in $\mathbb{T}$, $\phi_t$ and $\psi_t$ can be expressed by 
\begin{align}
&\phi_t(z)=\phi(W_tz)  \quad \forall z\in\Re^n\label{eq:phi_f}\\
& \psi_t(e)=  \psi(V_te) \quad \forall e\in\Re^{n_y}, \label{eq:psi_f}
\end{align}
where  $\phi:\Re^n\rightarrow\Re_{\geq 0}$ and $\psi:\Re^{n_y}\rightarrow\Re_{\geq 0}$ are two fixed loss functions and $\{W_t\}$ and $\{V_t\}$ are  two families of \textit{nonsingular} weighting matrices with appropriate dimensions.

\begin{definition}\label{def:est} 
Given a system $\Sigma$  such as the one in \eqref{eq:sys} and given an output measurement matrix $Y\in \Re^{n_y\times T}$, we define a state estimator to be a set-valued map $\mathcal{E}:\Re^{n_y \times T} \rightarrow \mathcal{P}(\Re^{n\times T})$ which maps  $Y$ to a subset of the space of possible trajectories of the system.    
\end{definition}
We consider a class of state estimators defined by
\begin{equation}\label{eq:def_est}
 \mathcal{E}(Y)=\argmin_{Z\in\Re^{n\times T}}V_{\Sigma}(Y,Z).
\end{equation}
As such the estimator $\mathcal{E}$ is well-defined if for any fixed $Y$,  $V_{\Sigma}(Y,Z)$ admits a non empty minimizing set, that is, if there exists at least one $Z^\star$ such that $V_{\Sigma}(Y,Z)\geq V_{\Sigma}(Y,Z^\star)$ for all $Z\in \Re^{n\times T}$. To ensure this property we will need to put an observability assumption on the system whose state is being estimated and require some further properties on the loss functions $\phi$ and $\psi$ entering in the definition of the objective function $V_\Sigma$. 

\unskip
\subsection{Well-definedness of the estimator}
Let us start by stating the properties required for the loss functions involved in the definition of $V_\Sigma$. 
Due to the multiple usages that will be made of  these properties, it is  convenient to state them for a generic  loss function defined on a set of matrices (of which vectors constitute a special case). 
Throughout this paper, a loss function is a positive function $\xi:\Re^{a\times b} \rightarrow\Re_{\geq 0}$ which will be required to satisfy a subset (depending of the specific usage)  of the following  properties: 
\begin{enumerate}[label=(P{\arabic*})]\topsep=0pt
\item Positive definiteness: $\xi(0)=0$ and $\xi(Z)>0$ for all $Z\neq 0$\label{prop:prem}
\item Continuity: $\xi$ is continuous \label{prop:2}
\item Symmetry:  $\xi(-Z)=\xi(Z)$ for all $Z\in\Re^{a\times b}$\label{prop:3}
\item \label{prop:4} {Generalized Homogeneity (GH)}: There exists a $\mathcal{K}_{\infty}$ function $q:\Re_{\geq 0} \rightarrow \Re_{\geq 0}$  such that for all $\lambda\in\Re_*$ and for all $Z\in\Re^{a\times b}$,
\begin{equation}\label{eq:homog_gen}
\xi(Z)\geq q\left(\dfrac{1}{|\lambda |}\right) \xi(\lambda Z).
\end{equation}
\end{enumerate}
\begin{enumerate}[label=(P{\arabic*})]\topsep=0pt \setcounter{enumi}{4}
\item \label{prop:der} {Generalized Triangle Inequality (GTI):} There exists a positive real number $\gamma_{\xi}$ such that for all $Z_1$, $Z_2$ in $\Re^{a\times b}$
\begin{equation}\label{eq:tri_gen}
\xi(Z_1-Z_2)\geq \gamma_{\xi}\xi(Z_1)-\xi(Z_2). 
\end{equation}
\end{enumerate}

\noindent It can be usefully observed, for the future developments, that \eqref{eq:tri_gen} can be equivalently written as 
$\xi(Z_1+Z_2)\leq \gamma_{\xi}^{-1}\xi(Z_1)+\gamma_{\xi}^{-1}\xi(Z_2). $

\paragraph{Examples of loss functions}
Note that norms on $\Re^{a\times b}$ satisfy naturally the properties~\ref{prop:prem}--\ref{prop:der} with $q:\lambda \mapsto \lambda$ and $\gamma_{\xi}=1$, hence  yielding the classic homogeneity property and triangle inequality. It can also be checked that functions $\xi$ of the form $\xi(Z)=\left\|Z\right\|^p$ with $p>0$, fully qualify as loss functions in the sense that they fulfill all the properties \ref{prop:prem}--\ref{prop:der}. In this case, $\gamma_{\xi}$ in \eqref{eq:tri_gen} can be taken equal to $2^{1-1/p}$ if $0<p\leq 1$ and  $2^{1-p}$ otherwise. Lastly we note that if $\ell:\Re^{a\times b}\rightarrow \Re_{\geq 0}$ satisfies \ref{prop:prem}--\ref{prop:3}  and \ref{prop:der}, then so does the  function $\xi$ defined by $\xi(Z)=1-e^{-\ell(Z)}$ (see Lemma \ref{lem:exponential-cost-P1-P5} in the appendix). Similarly, saturated functions of the form $\xi(Z)=\min(\ell(Z),R_0)$ for some $R_0>0$ satisfy  \ref{prop:prem}--\ref{prop:3}  and \ref{prop:der}.   
 In the case of convex functions, a link can be established between~\ref{prop:4} and~\ref{prop:der}. 
\begin{lem}[\cite{kircher_analysis_2020}]\label{lem:cvx}
If $\xi : \Re^{a\times b}\rightarrow \Re_{\geq 0}$ is convex and satisfies property~\ref{prop:4} with a $\mathcal{K}_{\infty}$ function $q$, then it also satisfies~\ref{prop:der} with  $\gamma_{\xi}=2q(1/2)$.
\end{lem}

Observe that quadratic functions $\xi :\Re^{a\times b} \rightarrow \Re_{\geq 0}$ of the form $\xi(Z)=\trace(Z^\top QZ)$ with $Q\in\Re^{a\times a}$ being a positive definite matrix and $\trace$ referring to the trace of a matrix, satisfy properties~\ref{prop:prem}--\ref{prop:4} with a $\mathcal{K}_\infty$ function $q:\lambda \mapsto \lambda^2$. Since such functions are convex,  it follows from Lemma \ref{lem:cvx} above that they also verify~\ref{prop:der} for $\gamma_{\xi}=2q(1/2)=1/2$.
\begin{rem}
In virtue of \eqref{eq:phi_f}-\eqref{eq:psi_f}, the families of functions $\left\{\phi_t\right\}$ and $\left\{\psi_t\right\}$ satisfy \ref{prop:prem}--\ref{prop:der} whenever $\phi$ and $\psi$ satisfy \ref{prop:prem}--\ref{prop:der}. 
\end{rem}
We now recall from \cite{kircher_analysis_2020} a technical lemma which will play a fundamental role in analyzing the properties of the estimator \eqref{eq:def_est}. In particular, our proof of well-definedness relies on this lemma. 
\begin{lem}[Lower Bound of a loss function]\label{lem:lb}
Let $\xi:\Re^{a\times b}\rightarrow\Re_{\geq 0}$ be a function which has properties \ref{prop:prem}--\ref{prop:2} and~\ref{prop:4} with a $\mathcal{K}_\infty$ function $q$. Then, for all norm $\lVert \cdot \rVert$ on  $\Re^{a\times b}$,
\begin{equation}\label{eq:fl_r}
\xi(Z)\geq Dq(\lVert Z \rVert) \quad \forall Z\in \Re^{a\times b}
\end{equation}
where
\begin{equation}\label{eq:def_d}
D=\min_{\lVert Z \rVert=1} \xi(Z)>0.
\end{equation}
\end{lem}
\begin{proof}
We start by observing that the unit hypersphere $\mathcal{S}=\left\{Z\in \Re^{a\times b}: \left\|Z\right\|=1\right\}$ is a compact set in the topology induced by the norm $\left\|\cdot\right\|$.  By the extreme value theorem, $\xi$ being continuous, admits necessarily a minimum value  on $\mathcal{S}$, i.e., there is $Z^\star\in \mathcal{S}$ such that $\xi(Z)\geq D\triangleq \xi(Z^\star)>0$ for all $Z\in \mathcal{S}$. For any nonzero $Z\in \Re^{a\times b}$, $\dfrac{Z}{\left\|Z\right\|}\in \mathcal{S}$ so that  $\xi(\dfrac{Z}{\left\|Z\right\|})\geq D$. On the other hand, by the relaxed homogeneity of $\xi$, 
$$\xi(Z)\geq q(\left\|Z\right\|) \xi(\dfrac{Z}{\left\|Z\right\|})\geq D q(\left\|Z\right\|).$$ Moreover, this inequality holds for $Z=0$. It therefore holds true for any $Z\in \Re^{a\times b}$. 
\end{proof}

\begin{prop}[Well-definedness of the estimator]\label{prop:Well-definedness}
Let the loss functions $\phi$ and $\psi$ in \eqref{eq:phi_f}-\eqref{eq:psi_f} satisfy properties \ref{prop:prem}--\ref{prop:der} 
 and assume that the LTV system \eqref{eq:sys} is observable on $\interval{0}{T-1}$ in the sense that the observability matrix 
%
\begin{equation}\label{eq:Obsv-Matrix}
\mathcal{O}_{0,T-1}\triangleq \begin{pmatrix}
(C_{0})^\top& 
(C_{1}A_{0})^\top&
\cdots&
(C_{T-1}A_{T-2}\ldots A_1 A_{0})^\top
\end{pmatrix}^\top
\end{equation}
has full column rank.    
Then the estimator \eqref{eq:def_est} is well-defined, \textit{i.e.}, the objective function $V_\Sigma(Y,\cdot)$ attains its minimum for any fixed $Y$.   
\end{prop}
Hence, the condition of the proposition guarantees that $\mathcal{E}(Y)$ is non empty for all $Y\in \Re^{n_y\times T}$. Before proving this  result, we first make the following observation.  

\begin{lem}[Equivalent condition of Observability]\label{lem:obs} 
Consider the objective function $V_\Sigma$ defined in \eqref{eq:V} where $\left\{(\phi_t,\psi_t)\right\}$ are defined as in \eqref{eq:phi_f}-\eqref{eq:psi_f} with $\phi$ and $\psi$ satisfying \ref{prop:prem}--\ref{prop:4}. 
Then the  following two statements are equivalent:
\begin{enumerate}
\item[(i)] \label{lo_ii} The system is observable on the time interval $\interval{0}{T-1}$.
\item[(ii)] \label{lo_i} There exists  a $\mathcal{K}_{\infty}$ function $q$ such that for all $Z~=~\begin{pmatrix}z_{0} & z_{1} &\ldots &z_{T-1}\end{pmatrix}$ in $\Re^{n\times T}$,
\begin{equation}\label{eq:lem_lo_i}                                                                                  
V_\Sigma(0,Z)\geq q(\lVert z_{0}\rVert)
\end{equation}
\end{enumerate}
\end{lem}
A proof of this lemma is reported in Appendix~\ref{app:obs_proof}. 
The function $q$ can be interpreted here as a gain function which measures how much the system is observable with regards to the two families $\{\phi_t\}$ and $\{\psi_t\}$: the more the system is observable, the more $q$ amplifies its argument magnitude, making different trajectories more discernible. 

\noindent\textit{Proof of Proposition \ref{prop:Well-definedness}:}
The idea of the proof is to show that $V_{\Sigma}(Y,\cdot)$ is \textit{coercive} (\textit{i.e.}, continuous and radially unbounded) for any given $Y$ and then apply a result\footnote{Note that radial unboundedness is equivalent to  level-boundedness in the terminology of \cite{Rockafellar05-Book}. } in \cite[Thm 1.9]{Rockafellar05-Book} to conclude on the attainability of the infimum (which certainly exists since $V_{\Sigma}(Y,\cdot)$ is a positive function). Clearly, $V_{\Sigma}(Y,\cdot)$ is continuous as a consequence of  $\phi$ and $\psi$ being continuous by assumption (see property \ref{prop:2}). We then just need to  prove the radial unboundedness of $V_{\Sigma}(Y,\cdot)$, \textit{i.e.},  
$\lim_{\left\|Z\right\|\rightarrow +\infty}V_{\Sigma}(Y,Z)=+\infty$ for an arbitrary norm $\left\|\cdot\right\|$ on the $Z$-space and for all fixed $Y$. Since $\psi$ satisfies property \ref{prop:der}, there exists a constant $\gamma_\psi>0$ such that 
$\psi_t(y_t-C_tz_t)\geq \gamma_\psi \psi_t(C_tz_t)-\psi_t(y_t)$. Applying this property leads naturally to 
$$V_{\Sigma}(Y,Z)\geq F(Z)-  \sum_{t\in\mathbb{T}}\psi_t(y_t),$$
where 
\begin{equation}
	F(Z)=\lambda \sum_{t\in\mathbb{T}'}\phi_t(z_{t+1}-A_tz_t)+\gamma_\psi \sum_{t\in\mathbb{T}}\psi_t(C_tz_t).  
\end{equation}
It can then be shown (following a similar reasoning as in Appendix \ref{app:obs_proof}), under the observability assumption,  that $F$ satisfies the conditions of Lemma \ref{lem:lb}. It follows that for any norm $\left\|\cdot\right\|$ on $\Re^{n\times T}$,  there exists a  $\mathcal{K}_\infty$ function $q$ such that 
$$ F(Z)\geq q(\left\|Z\right\|).$$ 
Combining this with the inequality above, we obtain that 
$$V_{\Sigma}(Y,Z)\geq q(\left\|Z\right\|)-  \sum_{t\in\mathbb{T}}\psi_t(y_t)$$
which implies the radial unboundedness of $V_{\Sigma}(Y,\cdot)$ for any fixed $Y$. Hence the estimator  \eqref{eq:def_est} is well-defined as stated. 
\qed

As it turns out from Proposition \ref{prop:Well-definedness}, observability of system \eqref{eq:sys} and properties \ref{prop:prem}--\ref{prop:4} imposed on $\phi$ and $\psi$ ensure that $\mathcal{E}(Y)$ is a non empty set for any given $Y$. We then call any member $\hat{X}=\begin{pmatrix}
\hat{x}_0 & \hat{x}_1 & \dots & \hat{x}_{T-1}\end{pmatrix}$ of $\mathcal{E}(Y)$, an estimate of the state trajectory $X$ of system \eqref{eq:sys} on the time interval $\mathbb{T}$. In particular, $\hat{x}_t$ is called an estimate of the state $x_t$ at time $t\in \mathbb{T}$.

To conclude this section, note that the definition  of the estimator in \eqref{eq:def_est} does not require any convexity assumption on the objective function $V_\Sigma$.  Hence the theoretical analysis to be presented in the next sections does not make use of  convexity either. However, we  may prefer in practice  to select convex loss functions $\phi$ and $\psi$. In effect, the elements of $\mathcal{E}(Y)$ are not necessarily expressible through an explicit formula. So, in practice one would resort instead to numerical solvers to approach the solution of the underlying optimization problem. And the numerical search process is known to be more efficient when $V_\Sigma(Y,Z)$ is a convex function of $Z$ \cite{Boyd04-Book,grant_cvx_2017}.
Nevertheless it is fair to recognize that nonconvex optimization methods can be  successfully implemented as well though with less theoretical guarantees of reaching global optimality with general purpose solvers. 


\section{The resilience property of the proposed class of estimators}\label{part:iv}

In this section, we prove that the state estimator proposed in \eqref{eq:def_est}  possesses the resilience property under some conditions. More specifically, our main result states that  the estimation error  $\hat{X}-X$, \textit{i.e.}, the difference between the real state trajectory and the estimated one, is upper bounded by a bound  which does not depend on the amplitude of the outliers contained in $\left\{f_t\right\}$ provided that the number of such outliers is below some threshold. 

\subsection{Definition of the resilience of an estimator}\label{subsec:def-resilience}
\noindent Let us start with a formal definition of the resilience property for a state estimator of the form \eqref{eq:def_est}. 
For this purpose, let $Y^\star=\begin{pmatrix} y_0^\star & \ldots,y_{T-1}^\star\end{pmatrix}$ denote the noise-free  output matrix of \eqref{eq:sys}, i.e., the output defined by $y_t^\star=C_tx_t^\star$, $t=0,\ldots,T-1$, with $x_{t+1}^\star=A_tx_{t}^\star$ and  $x_0^\star=x_0$ ($x_0$ being the true initial state of \eqref{eq:sys}). Let $\mathcal{F}_r$, a subset of $\Re^{n_y\times T}$ containing $0$, denote a matrix of measurement noise components. 
\begin{definition}[Resilience of an estimator]  \label{def:resilience}
The set-valued estimator $\mathcal{E}$ defined in \eqref{eq:def_est}  is called resilient against the set $\mathcal{F}_r$ of measurement noise if there exists a $\mathcal{K}_\infty$ function $g$ such that, when  the process noise $\left\{w_t\right\}$ is zero,  it holds that for any measurement noise matrix $F\in \Re^{n_y\times T}$, 
	\begin{equation}\label{eq:resilience-def}
		\|\hat{X}-X\| \leq g\big(\inf_{\Omega\in \mathcal{F}_r}d(F-\Omega)\big) \quad \forall \hat{X}\in \mathcal{E}(Y^\star+F)
	\end{equation}
with $\|\cdot\|$ denoting some norm, $d: \Re^{n_y\times T}\rightarrow\Re_{\geq 0}$ being a function subject to \ref{prop:prem}-\ref{prop:der}. Hence  $\inf_{\Omega\in \mathcal{F}_r}d(F-\Omega)$  denotes some pseudo-distance from $F$ to the set  $\mathcal{F}_r$.   
\end{definition}
Since $0\in \mathcal{F}_r$, a consequence of property \eqref{eq:resilience-def} is that  $\mathcal{E}(Y^\star)=\left\{X\right\}$ which follows from \eqref{eq:resilience-def} for $F=0$. This  fact expresses correctness of the estimator in a nominal situation, i.e., its ability to recover the true state matrix $X$ in the absence of any uncertainty in the a priori known model. Indeed this condition is guaranteed to hold if the system $\Sigma$ is observable over the considered observation time horizon $T$. 
Another key implication of  condition \eqref{eq:resilience-def} is that   the estimation error associated with a  resilient estimator is totally insensitive  to any measurement noise matrix $F$ which lies in $\mathcal{F}_r$, that is, $\mathcal{E}(Y^\star+F)=\left\{X\right\}$ for all measurement noise $F\in \mathcal{F}_r$. Throughout this paper, we consider a set   $\mathcal{F}_r$ defined as follows.  For $F=\begin{pmatrix} f_0 & \cdots & f_{T-1}\end{pmatrix}\in \Re^{n_y\times T}$, let $\mathbb{T}_{0}^c(F)=\left\{t\in \mathbb{T}:\psi_t(f_{t})>0 \right\}$ and $\mathbb{T}_{0}(F)=\left\{t\in \mathbb{T}:\psi_t(f_{t})=0\right\}$. For $r$ a positive integer,  
define  $\mathcal{F}_r$ to be  the set of matrices in $\Re^{n_y\times T}$ having at most $r$ nonzero columns, i.e., 
\begin{equation}\label{eq:Fr}
	\mathcal{F}_r=\left\{F: |\mathbb{T}_{0}^c(F)|\leq r\right\}.
\end{equation}
For the need of making explicit the resilience property (in the results to be presented) with respect to the set $\mathcal{F}_r$, we will need the following lemma. 
\begin{lem}\label{lem:Fr}
 Consider the set $\mathcal{F}_r$ of measurement noise matrix defined in \eqref{eq:Fr} and select a (pseudo distance) function $d:\Re^{n_y\times T}\rightarrow\Re_{\geq 0}$ defined by $d(F)= \sum_{t\in \mathbb{T}} \psi_t(f_t)$ with $\psi_t$ a function defined as in \eqref{eq:psi_f} and having the properties \ref{prop:prem}-\ref{prop:der}. Then
$\inf_{\Omega\in \mathcal{F}_r}d(F-\Omega)$ is equal to the sum of the $T-r$ smallest terms in $\left\{\psi_t(f_t):t\in \mathbb{T} \right\}$. 
\end{lem}
\begin{proof}
Let $I_r^c(F)$ denote the index set of the $r$ largest entries of the vector  $\begin{pmatrix}\psi_0(f_0) & \ldots & \psi_{T-1}(f_{T-1})\end{pmatrix}$ and $I_r(F)$ denote the index set of its $T-r$ smallest entries. Then, with the notation $\Omega=\begin{pmatrix}\omega_0 & \cdots & \omega_{T-1}\end{pmatrix}$,  
$$\begin{aligned}
	\inf_{\Omega\in \mathcal{F}_r}d(F-\Omega)& =\inf_{\Omega\in \mathcal{F}_r}\sum_{t\in \mathbb{T}}\psi_t(\omega_t-f_t)\\
	& = \inf_{\Omega\in \mathcal{F}_r}\Big[\!\!\sum_{t\in I_r(F)}\psi_t(\omega_t-f_t)+\!\!\sum_{t\in I_r^c(F)}\psi_t(\omega_t-f_t)\Big]\\
	& = \inf_{\substack{\Omega\in \mathcal{F}_r\\ \mathbb{T}_0(\Omega)=I_r(F)}}\Big[\!\sum_{t\in I_r(F)}\psi_t(f_t)+\!\!\!\sum_{t\in I_r^c(F)}\psi_t(\omega_t-f_t)\Big]\\
	& = \sum_{t\in I_r(F)}\psi_t(f_t)
\end{aligned} $$
where the notation $\mathbb{T}_0(\Omega)$ is defined in the lines preceding Eq. \eqref{eq:Fr}. 
The infimum is reached here for $\Omega\in \mathcal{F}_r$ such that $\omega_t=0$ $\forall t\in I_r(F)$ and $\omega_t=f_t$ $\forall  t\in I_r^c(F)$. Hence $\inf_{\Omega\in \mathcal{F}_r}d(F-\Omega)$ is, as claimed, the sum of the $T-r$ smallest values among $\left\{\psi_t(f_t):t\in \mathbb{T}\right\}$. 
\end{proof}

We will also introduce in the sequel a notion of \textit{approximate resilience} of $\mathcal{E}$. This terminology refers to Definition \ref{def:resilience} when the right hand side of \eqref{eq:resilience-def} is modified as 
$ g\big(\inf_{\Omega\in \mathcal{F}_r}d(F-\Omega)+\delta\big)$ with $\delta$ some nonnegative real number. 
 %
	

\subsection{Some notational conventions for the analysis}
For convenience, let us introduce a few more notations. 
Let $\Phi:\Re^{n\times T}\rightarrow\Re_{\geq 0}$ and $\Psi_{\mathbb{T}}:\Re^{n_y\times T}\rightarrow\Re_{\geq 0}$ be defined by 
\begin{align}
&\Phi(Z)=\sum_{t\in\mathbb{T}'}\phi_t(z_{t+1}-A_tz_t)\\
&\Psi_{\mathbb{T}}(Z)=\sum_{t\in\mathbb{T}}\psi_t(C_tz_t) \label{eq:Psi-T}
\end{align}
We also  introduce the partial cost function $\Psi_{\mathbb{I}}$ defined for any $\mathbb{I}\subset \mathbb{T}$ by  $\Psi_{\mathbb{I}}(Z)=\sum_{t \in \mathbb{I}} \psi_t(C_tz_t)$. We will assume throughout the paper that the loss functions $\phi$ and $\psi$ satisfy a subset of the properties  \ref{prop:prem}--\ref{prop:der}  and in particular, when they are required to satisfy the GTI~\ref{prop:der},  we will denote the associated positive constants with $\gamma_\phi$ and $\gamma_\psi$ respectively. 
Finally, let us pose
\begin{equation}\label{eq:H(Z)}
	H_\Sigma(Z)= \lambda \gamma_\phi \Phi(Z)+ \gamma_\psi\Psi_{\mathbb{T}}(Z).  
\end{equation}

We will organize the resilience analysis for the estimator \eqref{eq:def_est} along two cases: first, the scenario where the gross error vector sequence $\left\{s_t\right\}$ in \eqref{eq:f=s+v} is block-sparse in time and then the situation where it is both componentwise and temporally sparse. To be more precise, if we denote with $S\in \Re^{n_y\times T}$ the matrix formed from the sequence $\left\{s_t: t\in \mathbb{T}\right\}$, then the first case refers to columnwise block-sparsity of $S$ while the second is related to an entrywise sparsity. Note that the two cases coincide when the system of interest is  single-input single-output (SISO).

\subsection{Resilience to intermittent timewise block-sparse errors}
We start by introducing the concept of $r$-Resilience index of an estimator such as the one in \eqref{eq:def_est}, a measure which depends of the system matrices, the structure of the performance function $V_\Sigma$ and on the loss functions $\phi$ and $\psi$. 

\begin{definition}\label{def:pr}
Let   $r$ be a nonnegative integer. Assume that the system  $\Sigma$ in \eqref{eq:sys} is  observable on $\interval{0}{T-1}$.  We then define the $r$-Resilience index of the estimator $\mathcal{E}$ in \eqref{eq:def_est} (when applied to $\Sigma$) to be the real number $p_r$ given by 
\begin{equation}\label{eq:def_pr}
p_r=\sup_{\substack{Z\in \Re^{n\times T}\\ Z\neq 0}}\sup_{\substack{\mathbb{I}\subset \mathbb{T}\\|\mathbb{I}|=r}}\dfrac{\Psi_{\mathbb{I}}(Z)}{H_\Sigma(Z)}
\end{equation}
where $H_\Sigma$ is as defined in~\eqref{eq:H(Z)}. The supremum is taken here over all nonzero $Z$ in $\Re^{n\times T}$ and over all subsets  $\mathbb{I}$ of $\mathbb{T}$ with cardinality equal to $r$. 
\end{definition}
\noindent The index $p_r$ can be interpreted as a quantitative measure of the observability of the system $\Sigma$. The observability is needed here to ensure that the denominator $H_\Sigma(Z)$ of \eqref{eq:def_pr} is different from zero whenever $Z\neq 0$.  Furthermore, it should be remarked that $\Psi_{\mathbb{I}}(Z)\leq H_\Sigma(Z)$ for any $\mathbb{I}\subset \mathbb{T}$,  which implies that the defining suprema of  $p_r$ are well-defined. Note that $p_r$ is an increasing function of $r$ and satisfies $p_0=0$ and $p_T=1$. More discussions on the numerical evaluation of $p_r$ are deferred to Section \ref{sec:Numerical-Eval}.

In order to state the resilience result for the estimator \eqref{eq:def_est} when applied to system $\Sigma$,  let us introduce a last notation to be used in the analysis. 
Let $\varepsilon\geq 0$ be a given number. For any admissible  sequence $\left\{f_t\right\}_{t\in \mathbb{T}}$ in \eqref{eq:sys}, we can  split the time index set $\mathbb{T}$ into two disjoint label sets,  
\begin{equation}\label{eq:def_teps}
	\mathbb{T}_\varepsilon=\left\{t\in \mathbb{T}: \psi_t(f_{t})\leq \varepsilon\right\},
\end{equation}
indexing those $f_{t}$ which are upper bounded by $\varepsilon$ and  $\mathbb{T}_\varepsilon^c=\left\{t\in \mathbb{T}: \psi_t(f_t)> \varepsilon\right\}$ indexing those $f_{t}$ which are possibly unbounded. It is important to keep in mind that $\varepsilon$ is just a parameter for decomposing the noise sequence in two parts in view of the analysis (and not necessarily a bound on elements of the sequence $\left\{\psi(f_{t})\right\}$). 
For example,  taking $\varepsilon=0$ would be appropriate for analyzing the properties of the estimator when $f_{t}$ is strictly sparse and each of its nonzero elements is treated as an outlier.

\begin{theorem}[Upper bound on the estimation error]\label{th:lb}
Consider the system $\Sigma$ defined by~\eqref{eq:sys} with output $Y$ together with the state estimator \eqref{eq:def_est} in which the loss functions $\phi$ and $\psi$ are assumed to obey \ref{prop:prem}--\ref{prop:der}. Denote with $\gamma_\phi$ and $\gamma_\psi$ the constants associated with  the GTI~\ref{prop:der}  and $q_\phi$ and $q_\psi$ the $\mathcal{K}_\infty$ functions associated with the GH~\ref{prop:4} \color{black} for $\phi$ and $\psi$ respectively. 
Let  $\varepsilon\geq 0$ and set  $r=|\mathbb{T}_{\varepsilon}^c|$. \\
If the system is observable on $\interval{0}{T-1}$ and $p_r< 1/(1+\gamma_\psi)$,  then for any norm  $\lVert \cdot \rVert$ on $\Re^{n\times T}$,
\begin{equation}\label{eq:th_lb}
\lVert \hat{X}-X\rVert \leq h\Big(\dfrac{2\beta_\Sigma(\varepsilon)}{D\big[1-(1+\gamma_\psi)p_r\big]}+\delta(\varepsilon)\Big) \quad \forall  \hat{X}\in \mathcal{E}(Y)
\end{equation} 
with $X$ denoting the true state matrix from \eqref{eq:sys},  $D=\min_{\left\|Z\right\|=1}H_{\Sigma}(Z)>0$ and $\beta_\Sigma(\varepsilon)$, $\delta(\varepsilon)$  and $h$ \color{black} being defined by 
\begin{align}
	& \beta_\Sigma(\varepsilon)=\lambda\sum_{t\in\mathbb{T}'}\phi_t (w_t)+\sum_{t\in \mathbb{T}_\varepsilon}\psi_t(f_t), \label{eq:beta} \\
		& \delta(\varepsilon) = \dfrac{1-\gamma_\psi}{D\left[1-(1+\gamma_\psi)p_r\right]}\sum_{t\in \mathbb{T}_\varepsilon^c} \psi_t(f_t) \label{eq:delta}\\
		 & h(\alpha)=\max\left\{q_\phi^{-1}(\alpha),q_\psi^{-1}(\alpha)\right\}, \alpha \in \Re_{\geq 0}\label{eq:def_h}
\end{align}
\end{theorem}

\begin{proof} Let $\hat{X}$ in $\mathcal{E}(Y)$.  By definition of $\mathcal{E}$ in \eqref{eq:def_est}, we have $V_{\Sigma}(Y,\hat{X})\leq V_{\Sigma}(Y,X)$, which gives explicitly 
\begin{multline}\label{eq:V(Xhat)<V(X)}
\lambda \sum_{t\in\mathbb{T'}}\phi_t(\hat{x}_{t+1}-A_t\hat{x}_t)+\sum_{t\in\mathbb{T}}\psi_t(y_t-C_t\hat{x}_t) \\ \leq \lambda \sum_{t\in\mathbb{T}'}\phi_t(w_t)+\sum_{t\in\mathbb{T}}\psi_t(f_t)
\end{multline}
Using the fact that $x_{t+1}=A_tx_t+w_t$ from \eqref{eq:sys} and applying the GTI and the symmetry properties of $\phi_t$, we can write
\begin{align*}
\phi_t (\hat{x}_{t+1}-A_t\hat{x}_{t})
& = \phi_t(\hat{x}_{t+1}-x_{t+1}-A_t(\hat{x}_t-{x}_t)+w_t)  \\
&\geq\gamma_\phi \phi_t(e_{t+1}-A_te_t)-\phi_t(w_t) 
\end{align*}
with $e_t=\hat{x}_t-{x}_t$. It follows that the first term on the left hand side of \eqref{eq:V(Xhat)<V(X)} is lower bounded as follows
\begin{equation}\label{eq:INEQ-I}
	\lambda \sum_{t\in\mathbb{T'}}\left[\gamma_\phi \phi_t(e_{t+1}-A_te_t)-\phi_t(w_t)\right]\leq \lambda \sum_{t\in\mathbb{T'}}\phi_t(\hat{x}_{t+1}-A_t\hat{x}_t). 
\end{equation}
Similarly, by making use of \eqref{eq:sys}, observe that $\psi_t(y_t-C_t\hat{x}_t)=\psi_t(f_t-C_te_t)$. We now apply the GTI and symmetry of $\psi_t$ in two different ways  depending on whether $t$ belongs to $\mathbb{T}_\varepsilon$ or $\mathbb{T}_\varepsilon^c$: 
\begin{align*}
\forall t\in \mathbb{T}_\varepsilon,\:&& &\psi_t(y_t-C_t\hat{x}_t)\geq  \gamma_\psi \psi_t(C_te_t)-\psi_t(f_t)\\
\forall t\in \mathbb{T}_\varepsilon^c,\:&& &\psi_t(y_t-C_t\hat{x}_t)\geq \gamma_\psi \psi_t(f_t)-\psi_t(C_te_t)
\end{align*}
These inequalities imply that the second term  on the left hand side of \eqref{eq:V(Xhat)<V(X)} is lower bounded as follows
\begin{equation}\label{eq:INEQ-II}
	\begin{aligned}
		\sum_{t\in \mathbb{T}_\varepsilon}\left[\gamma_\psi \psi_t(C_te_t)-\psi_t(f_t)\right]&+\sum_{t\in \mathbb{T}_\varepsilon^c}\left[\gamma_\psi \psi_t(f_t)-\psi_t(C_te_t)\right] \\
		&\qquad \leq \sum_{t\in\mathbb{T}}\psi_t(y_t-C_t\hat{x}_t)	
	\end{aligned} 
\end{equation}
Combining \eqref{eq:V(Xhat)<V(X)}, \eqref{eq:INEQ-I} and \eqref{eq:INEQ-II} gives 
$$
\begin{aligned}
&\lambda \gamma_\phi \sum_{t\in\mathbb{T}'} \phi_t(e_{t+1}-A_te_t)+\gamma_\psi\sum_{t\in \mathbb{T}}\psi_t(C_te_t)\\
& \hspace{3cm} -(1+\gamma_\psi)\sum_{t\in \mathbb{T}_\varepsilon^c}\psi_t(C_te_t) \\
&\qquad \leq 2\Big(\lambda\sum_{t\in\mathbb{T}'} \phi_t (w_t)+\sum_{t\in \mathbb{T}_\varepsilon}\psi_t(f_t)\Big)+\sum_{t\in \mathbb{T}_\varepsilon^c}(1-\gamma_\psi)\psi_t(f_t)
\end{aligned}
$$
which, by using  \eqref{eq:Psi-T}, \eqref{eq:H(Z)}, \eqref{eq:beta}, can be written as 
$$
\begin{aligned}
&H_\Sigma(E)-(1+\gamma_\psi)\Psi_{\mathbb{T}_\varepsilon^c}(E)
\leq 2\beta_\Sigma(\varepsilon)+\sum_{t\in \mathbb{T}_\varepsilon^c}(1-\gamma_\psi)\psi_t(f_t)
\end{aligned}
$$
with $E=\begin{pmatrix}e_0 & e_1 & \cdots & e_{T-1}\end{pmatrix}$. 
As $\mathbb{T}_\varepsilon^c$ has $r$ elements, applying the definition of $p_r$ gives
\begin{equation}
\Psi_{\mathbb{T}_\varepsilon^c}(E)\leq p_r H_\Sigma(E)
\end{equation}
By the assumption that $p_r<1/(1+\gamma_\psi)$ we have that \\ $1-(1+\gamma_\psi)p_r>0$, and consequently, that 
\begin{equation}\label{eq:proof_i}
 H_{\Sigma}(E)\leq \dfrac{1}{1-(1+\gamma_\psi)p_r}\Big[2\beta_\Sigma(\varepsilon)+(1-\gamma_\psi)\sum_{t\in \mathbb{T}_\varepsilon^c} \psi_t(f_t)\Big]
\end{equation}
Given that the system is observable on $\interval{0}{T-1}$, it can be shown, thanks to Lemma~\ref{lem:fc} in the Appendix, that $H_\Sigma$ satisfies properties~\ref{prop:prem}--\ref{prop:4} (the proof of this is quite similar to that of Lemma \ref{lem:obs} in  Appendix~\ref{app:obs_proof}). We can therefore apply  Lemma~\ref{lem:lb} to conclude that for any norm $\lVert \cdot \rVert$ 
\begin{equation}\label{eq:proof_ii}
H_{\Sigma}(E)\geq Dq'(\lVert E\rVert)
\end{equation}
with  $D$ defined by $D=\min_{\lVert Z \rVert=1}H_\Sigma(Z)$  and $q'(\alpha)=\min\{q_\phi(\alpha),q_\psi(\alpha)\}$\color{black}. 
Finally, the result follows  by selecting $h$ to be $h=q'^{-1}$ with $q'^{-1}$ denoting the inverse of $q'$, which can be simplified to match its definition in~\eqref{eq:def_h}.   
\end{proof}

\paragraph{Strict resilience}
Now we state our (strict) resilience result as a consequence of Theorem \ref{th:lb}  when  the output error-measuring loss function $\psi$ satisfies the triangle inequality.  
\begin{coro}[Resilience property]
\label{coro:resilience}
Let the conditions of Theorem \ref{th:lb} hold with the additional requirement that $\gamma_\psi=1$.  
Then 
\begin{equation}\label{eq:resilience}
\lVert \hat{X}-X\rVert \leq h\Big(\dfrac{2\beta_\Sigma(\varepsilon)}{D\big(1-2p_r\big)}\Big) \quad \forall  \hat{X}\in \mathcal{E}(Y). 
\end{equation} 
\end{coro}
\begin{proof}
The proof is immediate by considering the bound  in \eqref{eq:th_lb} and observing that $\delta(\varepsilon)$  expressed in \eqref{eq:delta} vanishes when $\gamma_\psi=1$, hence eliminating completely the contribution of the extreme values of $\left\{f_t\right\}$ to the error bound. This gives immediately \eqref{eq:resilience}. It remains now to make clear that \eqref{eq:resilience} is consistent with the requirement \eqref{eq:resilience-def} of Definition \ref{def:resilience}. For this purpose note that the bound in \eqref{eq:resilience} can be written as $g(\beta_\Sigma(\varepsilon))$ with $g\in \mathcal{K}_\infty$ defined by $g(\alpha)=h\big(2\alpha/\big(D(1-2p_r)\big)\big)$. Moreover, $\beta_\Sigma(\varepsilon)$ reduces to $\sum_{t\in \mathbb{T}_\varepsilon}\psi_t(f_t)=\inf_{\Omega\in \mathcal{F}_r}d(F-\Omega)$ when the process noise is zero (see \eqref{eq:beta} and Lemma \ref{lem:Fr}). Hence, $\mathcal{E}$ qualifies, in the sense of Definition \ref{def:resilience}, as an estimator which is resilient to the set $\mathcal{F}_r$ of measurement noise defined in \eqref{eq:Fr}.  
\end{proof}
The resilience property of the estimator \eqref{eq:def_est} lies here in the fact that, under the conditions of Theorem \ref{th:lb} and Corollary \ref{coro:resilience}, the bound in \eqref{eq:resilience} on the estimation error does not depend on the magnitudes of the extreme values of the noise sequence $\left\{f_{t}\right\}_{t\in \mathbb{T}}$. Considering in particular the function $\beta_\Sigma(\varepsilon)$, we remark that it can be overestimated as follows  
\begin{equation}\label{eq:upper-bound-beta}
\beta_\Sigma(\varepsilon)\leq \lambda\sum_{t\in\mathbb{T}'}\phi_t(w_t)+|\mathbb{T}_\varepsilon|\varepsilon. 
\end{equation}
The first term on the left hand side of \eqref{eq:upper-bound-beta}  represents the uncertainty brought by the dense noise $\left\{w_t\right\}$ over the whole state trajectory. It is bounded since $\left\{w_t\right\}$ is bounded by assumption (see the description of the system in Section \ref{part:ii}). The second term is a bound on the sum of those instances of $f_{t}$ whose magnitude is smaller that $\varepsilon$. \\ 
Because $\beta_\Sigma$ is a function of $\varepsilon$, the bound in \eqref{eq:resilience}  represents indeed a family of bounds  parameterized by $\varepsilon$. Since $\varepsilon$ is a mere analysis device, a question would be how to select it for the analysis to achieve the smallest bound. Such favorable values, say $\varepsilon^\star$, satisfy  
$$\varepsilon^\star\in \argmin_{\varepsilon\geq 0} \left\{h\Big(\dfrac{2\beta_\Sigma(\varepsilon)}{D(1-2p_r)}\Big): r=|\mathbb{T}_\varepsilon^c|,\: p_r<1/2 \right\}.  $$

Another interesting point is that the inequality stated by Theorem~\ref{th:lb} holds for any norm $\left\|\cdot\right\|$ on $\Re^{n\times T}$. Note though that the value of the bound depends (through the parameter $D$) on the specific norm used to measure the estimation error.  Moreover, different choices of the performance-measuring norm will result in different geometric forms for the uncertain set, that is, the ball (in the chosen norm) centered at the true state with radius equal to the upper bound displayed in \eqref{eq:resilience}.

We also observe that  the smaller the parameter $p_r$ is, the tighter the error bound will be, which suggests that the estimator is more resilient when $p_r$ is lower. A similar reasoning applies to the number $D$ which is desired to be large here.  
These two parameters (i.e., $p_r$ and $D$) reflect properties of the system whose state is being estimated. They can be interpreted, to some extent, as measures of the degree of observability of the system. In conclusion, the estimator inherits partially its resilience property from characteristics of the system being observed. This is consistent with the well-known fact that the more observable a system is, the more robustly its state can be estimated from output measurements.

\paragraph{Approximate resilience}
As discussed above, the triangle inequality property of the loss function $\psi$ is fundamental for achieving \textit{strict resilience}. When $\psi$ does not satisfy this property (i.e., when $\gamma_\psi\neq 1$), the term $\delta(\varepsilon)$ in \eqref{eq:th_lb} is unlikely to vanish completely. However we can prevent it from growing excessively by an appropriate choice of $\psi$ in \eqref{eq:psi_f}. To see this, assume for example that $\psi$ is defined by $\psi(y)=1-e^{-\ell(y)}$. Then since $\psi(y)\leq 1$ for all $y$, $\delta(\varepsilon)$ saturates to a constant value regardless of how large the $f_t$ are for $t\in \mathbb{T}_\varepsilon^c$. On the other hand, this choice introduces a new technical challenge related to the fact that the function $q'$ in \eqref{eq:proof_ii} is no longer a $\mathcal{K}_\infty$ function but a bounded (saturated) function. Handling this situation will require some additional condition on the upper bound in \eqref{eq:proof_i}.   
To sum up, by selecting a saturated loss function for $\psi$, we obtain the following \textit{approximate resilience} result.   
\begin{coro}[Case where $\gamma_\psi\neq 1$]\label{coro:Approximate-Resilience}
Let the conditions of Theorem \ref{th:lb} hold.   
Assume further that the loss function $\psi$ in \eqref{eq:psi_f} is defined by $\psi(y)=1-e^{-\ell(y)}$ where $\ell:\Re^{n_y}\rightarrow \Re_{\geq 0}$ satisfies \ref{prop:prem}--\ref{prop:der}. In particular, assume that property \ref{prop:4} is satisfied by $\ell$ with a  $\mathcal{K}_\infty$ function  $q$ such that \eqref{eq:homog_gen} is an equality relation.
Also, let $\varepsilon\geq 0$ be such that
\begin{equation}\label{eq:Cond-Bound}
	{b}(\varepsilon)\triangleq \dfrac{2\beta_\Sigma(\varepsilon)+r(1-\gamma_\psi)r^o(\varepsilon)}{ D\left[1-(1+\gamma_{\psi})p_r\right]}< 1,
\end{equation}
where $r=r(\varepsilon)=|\mathbb{T}_\varepsilon^c|$ and  $r^o(\varepsilon)=\max_{t\in \mathbb{T}_\varepsilon^c}\psi_t(f_t)\leq 1$. Then   there exists a continuous and strictly increasing function $h_{\sat}:\interval{0}{1}\rightarrow \interval{0}{1}$ (obeying $h_{\sat}(0)=0$ and $h_{\sat}(1)=1$) such that for any norm  $\lVert \cdot \rVert$ on $\Re^{n\times T}$,
\begin{equation}\label{eq:approximate-resilience}
\lVert \hat{X}-X\rVert \leq h_{\sat}^{-1}\big(b(\varepsilon)\big) \quad \forall  \hat{X}\in \mathcal{E}(Y). 
\end{equation} 
with $D$ in \eqref{eq:Cond-Bound} defined as in the proof of Theorem \ref{th:lb} using the norm $\lVert \cdot \rVert$.
\end{coro}
\begin{proof}
That the particular function $\psi$ specified in the statement of the corollary satisfies the properties \ref{prop:prem}--\ref{prop:3} and \ref{prop:der} is a question which is fully answered by Lemma \ref{lem:exponential-cost-P1-P5} in Section \ref{subsec:appendix-Corollary} of the appendix.  Consequently, let us observe that the inequality \eqref{eq:proof_i} arising in the proof of Theorem \ref{th:lb}  still holds true here. As to \eqref{eq:proof_ii}, it also holds as well but with the slight difference that $q'$ is just a saturated function  in $\mathcal{K}_{\sat,1}$ (as defined in the notation section) with bounded range $\interval{0}{1}$. This results in fact from Lemma \ref{lem:exponential-cost-P1-P5}  and the proof of Lemma \ref{lem:lb}.  We can therefore write 
$$
\begin{aligned}
	q'(\left\|E\right\|)&\leq \dfrac{1}{D\left(1-(1+\gamma_\psi)p_r\right)}\Big[2\beta_\Sigma(\varepsilon)+(1-\gamma_\psi)\sum_{t\in \mathbb{T}_\varepsilon^c} \psi_t(f_t)\Big]\\
	&\leq b(\varepsilon)< 1
\end{aligned}
$$
with $q'\in \mathcal{K}_{\sat,1}$. 
  Note from the definition of the class $\mathcal{K}_{\sat,1}$, that $q'(\left\|E\right\|)<1$ implies that $\left\|E\right\|<1$ (since otherwise we would have $q'(\left\|E\right\|)= 1$). 
Letting $h_{\sat}$ be the restriction of such a function $q'$ on $\interval{0}{1}$, we have $q'(\left\|E\right\|)=h_{\sat}(\left\|E\right\|)\leq b(\varepsilon)$  with $h_{\sat}$ being invertible. 
 We can now apply $h_{\sat}^{-1}$ to each member of this inequality to reach the desired result since $b(\varepsilon)$ lies in the range of $h_{\sat}$. 
\end{proof}


\subsection{Resilience to attacks on the individual  sensors}
We now consider the situation where the matrix  $S\in \Re^{n_y\times T}$ formed from $\left\{s_t\right\}$ in \eqref{eq:f=s+v} may be sparse entrywise i.e., a relatively important fraction of the entries of $S$ are equal to zero\footnote{In this case, the sparsity is expressed in term of fraction of \textit{nonzero entries} in the matrix $S$ whereas in the timewise block-sparsity case, the sparsity level is measured in term of the fraction of \textit{nonzero columns} in $S$.}.
This case is relevant when any individual sensor may be faulty (or compromised by an attacker) at any time.  To address the resilient state estimation problem in this scenario, we select the loss functions $\psi_t$  to have a separable structure.  To be more specific, let $\psi_t$ be such that for any $e=\bbm e_1 & \cdots & e_{n_y}\eem\in \Re^{n_y}$
 \begin{equation}\label{eq:psi_dec}
\psi_t(e)=\sum_{i=1}^{n_y} \psi_{ti}(e_i)
\end{equation} 
where, consistently with \eqref{eq:psi_f},   $\psi_{ti}(e_i)=\psi_i^{\circ}(V_{ti} e_i)$ with $V_{ti}\in \Re_{>0}$ and  $\psi_i^{\circ}:\Re\rightarrow\Re_+$, $i=1,\ldots,n_y$, being some loss functions on $\Re$ enjoying the properties \ref{prop:prem}--\ref{prop:der}.  As in the statement of Corollary \ref{coro:resilience}, we shall require that $\gamma_{\psi_i^{\circ}}=1$. It follows that one can set $\psi_i^{\circ}$ to be the absolute value without loss of generality. Let therefore set $\psi_i^{\circ}(e_i)=|e_i|$ so that $\psi_{ti}(e_i)=\left|V_{ti}e_i\right|$ and 
 \begin{equation}\label{eq:psi-L1}
\psi_t(e)=\left\|V_te\right\|_1
\end{equation} 
with $V_t$ being a diagonal matrix having the $V_{ti}$, $i=1,\ldots,n_y$, on its diagonal. 

To state  the resilience property in this particular setting,  we partition the index set $\mathbb{T}\times \mathbb{S}$ of the entries of $S$ as 
\begin{equation}\label{eq:Lambda-c}
	\begin{aligned}
	{\Lambda}_{\varepsilon}&=\{(t,i)\in \mathbb{T}\times \mathbb{S}\: :\: \psi_{ti}(f_{ti})\leq \varepsilon\}\\
	{\Lambda}_{\varepsilon}^c&=\{(t,i)\in \mathbb{T}\times \mathbb{S}\: : \: \psi_{ti}(f_{ti})> \varepsilon\}
	\end{aligned}
\end{equation}
with $f_{ti}$ denoting the $i$-th entry of the vector $f_t\in \Re^{n_y}$.  Also, in order to account for the specificity of the new scenario, let us refine slightly the $r$-Resilience index  \eqref{eq:def_pr} to be 
\begin{equation}\label{eq:pr_norm1}
\tilde{p}_r=\sup_{\substack{Z\in \Re^{n\times T}\\Z\neq 0}}\:\sup_{\substack{\mathbb{I}\subset \mathbb{T}\times\mathbb{S}\\|\mathbb{I}|=r}}\dfrac{\sum_{(t,i)\in \mathbb{I}}\psi_{ti}(c_{ti}^\top z_t)}{H_\Sigma(Z)}
\end{equation}
where $H_\Sigma$ is defined as in \eqref{eq:H(Z)} from $\psi_t$ in \eqref{eq:psi-L1} and  $c_{ti}^\top$ is $i$-th row of the observation matrix $C_t$.
 The difference between $p_r$ in \eqref{eq:def_pr} and $\tilde{p}_r$ in \eqref{eq:pr_norm1}  resides in the index sets for counting possible error occurrences which are $\mathbb{T}$ and $\mathbb{T}\times \mathbb{S}$, respectively.  

With these notations, we can provide the following theorem which is the analog of Corollary \ref{coro:resilience} in the case where the disturbance matrix $S$ (see Eq. \eqref{eq:f=s+v}) is entrywise sparse. 
\begin{theorem}[Upper bound on the estimation error: Separable case]\label{th:norm1}
Consider the system $\Sigma$ defined by~\eqref{eq:sys} with output $Y$ together with the state estimator \eqref{eq:def_est} in which  $\phi$ is assumed to obey \ref{prop:prem}--\ref{prop:der} and $\psi$ is defined as in \eqref{eq:psi-L1}. 
Let  $\varepsilon\geq 0$ and set  $r=|\Lambda_{\varepsilon}^c|$ with $\Lambda_{\varepsilon}^c$ defined in \eqref{eq:Lambda-c}. \\
If the system is observable on $\interval{0}{T-1}$ and if $\tilde{p}_r< 1/2$,  then there exists a $\mathcal{K}_\infty$ function $\tilde{h}$ such that for all norm  $\lVert \cdot \rVert$ on $\Re^{n\times T}$, 
\begin{equation}\label{eq:th_ub_dm}
\lVert \hat{X}-X\rVert \leq \tilde{h}\left(\dfrac{2\tilde{\beta}_\Sigma(\varepsilon)}{\tilde{D}(1-2\tilde{p}_r)}\right) \quad \forall  \hat{X}\in \mathcal{E}(Y)
\end{equation} 
with $X$ denoting the true state matrix from \eqref{eq:sys} and $\tilde{\beta}_\Sigma(\varepsilon)$  defined by 
$$
\begin{aligned}
	&\tilde{\beta}_\Sigma(\varepsilon)=\lambda\sum_{t\in\mathbb{T}'}\phi_t (w_t)+\sum_{(t,i)\in \Lambda_\varepsilon}\psi_{ti}(f_{ti})
\end{aligned}
$$
$\tilde{D}$  and $\tilde{h}$ \color{black} are defined as in the statement of Theorem \ref{th:lb} with $H_\Sigma$ being constructed from $\psi$ in \eqref{eq:psi-L1}. 
\end{theorem}
To some extent, Theorem \ref{th:norm1} can be viewed as a special case of Theorem \ref{th:lb} in which the  function $\psi$ is taken to be the $\ell_1$ norm and the data set is modified to be $\mathbb{T}\times \mathbb{S}$. Hence the proof follows a similar line of arguments as that of Theorem \ref{th:norm1}. Again it is not hard to see that the result of Theorem \ref{th:norm1} implies the property of resilience with respect to the set $\mathcal{F}_r$ in \eqref{eq:Fr} of measurement noise in the sense of  Definition \ref{def:resilience} (see the proof of Corollary \ref{coro:resilience}).   

An interesting property of the estimator can be stated in the absence of dense noise, i.e., when only the sparse noise is active:
\begin{coro}\label{coro:Exact-Recovery}
Consider the system $\Sigma$  defined by~\eqref{eq:sys} and let $r=|\Lambda_0^c|$ (which means that we consider every nonzero occurrence of $f_{it}$ as an outlier by taking $\varepsilon=0$ in \eqref{eq:Lambda-c}). If the conditions of Theorem \ref{th:norm1} hold, $\tilde{p}_r<1/2$, and if  $w_t=0$ in \eqref{eq:sys} for all $t$, then the estimator defined by~\eqref{eq:def_est} retrieves exactly the state trajectory of the system, i.e., 
$\mathcal{E}(Y)=\left\{X\right\}$. 
\end{coro}

\begin{proof}
This follows directly from the fact that $\tilde{\beta}_\Sigma(0)=0$ in the case where there is no dense noise $w_t$ and $\varepsilon=0$. 
\end{proof}
\noindent Therefore, the estimator \eqref{eq:def_est} has the exact recoverability property, that is, it is able to recover exactly the true state of the system \eqref{eq:sys} when only the sparse noise is active in the measurement equation provided that the number $r=|\Lambda_0^c|$ of nonzero in the sequence $\left\{f_{ti}\right\}_{(t,i)\in \mathbb{T}\times \mathbb{S}}$ is small enough for $\tilde{p}_r$ to be less than $1/2$. According to our analysis, the number of outliers that can be handled by the estimator in this case can be underestimated by 
\begin{equation}
	\max\big\{r: \tilde{p}_r<1/2\big\}. 
\end{equation}

\section{A special variant of the estimator $\mathcal{E}$}\label{sec:Exact-Recoverability}
In this section, we consider a constrained reformulation of the estimator $\mathcal{E}$ defined in \eqref{eq:def_est}. As will be shown shortly, this reformulation  also enjoys the resilience property but under a condition which is more easily verifiable from a  numerical perspective. 
\color{black} 

We start by considering the  simple scenario where the process noise $w_t$ in \eqref{eq:sys} is identically equal to zero and  the sequence $\left\{f_t\right\}$ is sparse in the sense that its dense component $v_t$ displayed in \eqref{eq:f=s+v}  does not exist.  In this setting  we can obtain a more resilient (to sparse noise in the measurement) estimator than \eqref{eq:def_est} by making it aware of the absence of dense process noise. This can be achieved by contraining the searched state matrix to be in the set $\mathcal{Z}_\Sigma\subset \Re^{n\times T}$  defined by 
$$\mathcal{Z}_\Sigma=\left\{Z= \begin{pmatrix}z_0 & A_0z_0 & \cdots & A_{T-1}\cdots A_1 A_0 z_0 \end{pmatrix}:z_0\in \Re^n \right\} $$
of possible trajectories starting in any initial state $z_0\in \Re^n$. Following this idea, we consider the estimator $\mathcal{E}^\circ$  defined by 
$$ \mathcal{E}^\circ(Y)=\argmin_{Z \in \mathcal{Z}_\Sigma}V_{\Sigma}(Y,Z). $$
Then $\mathcal{E}^\circ(Y)$ can be rewritten more simply in the form  
\begin{equation}\label{eq:Constrained-Estimator}
	\begin{aligned}
		\mathcal{E}^\circ(Y)&=\Big\{Z= \begin{pmatrix}z_0 & A_0z_0 & \cdots & A_{T-1}\cdots A_1 A_0 z_0 \end{pmatrix}:\Big. \\
		&\hspace{3.5cm} \Big. z_0\in \argmin_{z\in\Re^{n}}V_\Sigma^\circ(Y,z) \Big\} 
	\end{aligned}
\end{equation}
where 
\begin{equation}\label{eq:V-Sigma-circ}
	V_\Sigma^\circ(Y,z)=\sum_{t\in\mathbb{T}}\psi_t(y_t- M_t z)
\end{equation}
with 
\begin{equation}\label{eq:Mt}
	M_t=C_tA_{t-1}\cdots A_1A_0
\end{equation}
for all $t\geq 1$ and $M_0=C_0$.  
Hence the estimation of the state trajectory reduces to estimating the initial state $x_0$. This can be viewed as a robust regression problem, like the ones discussed in \cite{Han19-TAC,bako_class_2017}. Generalizing a result in \cite{bako_class_2017}, we derive next a necessary and sufficient condition for exact recovery of the true state, which holds if and only if
$\argmin_{z\in\Re^{n}}V_\Sigma^\circ(Y,z) =\{x_0\}$ with $x_0$ being the exact initial state of the system $\Sigma$. 
To this end, we first introduce the concept of concentration ratio of a collection of matrices with respect to a loss function. A notational convention will be necessary for the statement of this property: for any subset $\mathbb{I}$ of $\mathbb{T}$, let 
\begin{equation}\label{eq:Psi-circ}
	\Psi_{\mathbb{I}}^\circ(z)=\sum_{t\in\mathbb{I}}\psi_t(M_t z). 
\end{equation}
\begin{definition}[$r$-th concentration ratio]\label{def:concentration-ratio}
Let   $\{\psi_t\}$ be a family of loss functions defined by~\eqref{eq:psi_f} in which $\psi$ is assumed to satisfy \ref{prop:prem}, \ref{prop:3} and \ref{prop:der} with constant $\gamma_\psi=1$. Let $M=\left\{M_t\right\}_{t\in \mathbb{T}}$ be a sequence of matrices such that the function $\Psi_{\mathbb{T}}^\circ$ defined in \eqref{eq:Psi-circ} is positive definite. We call $r$-th concentration ratio of $M$, the number defined by
\begin{equation}\label{eq:nuR-M}
\nu_r(M)=\sup_{\substack{z\in \Re^n\\ z\neq 0}}\sup_{\substack{\mathbb{I}\subset \mathbb{T}\\|\mathbb{I}|=r}}\dfrac{\displaystyle\Psi_{\mathbb{I}}^\circ(z)}{\Psi_{\mathbb{T}}^\circ(z)}
\end{equation}
\end{definition}
At a fixed $r$, $\nu_r(M)$ quantifies a genericity property for the sequence $M=\left\{M_t\right\}_{t\in \mathbb{T}}$. In view of the particular structure of the collection $M$ in \eqref{eq:Mt},  note that $\Psi_{\mathbb{T}}^\circ$ is positive definite whenever the system $\Sigma$ is observable on $\mathbb{T}$. Furthermore, $\nu_r(M)$ can be interpreted to some extent, as a quantitative measure of observability. It is indeed all the smaller as the system is strongly observable. To see this, recall from Lemma \ref{lem:obs}  that if the system is observable on $\interval{0}{T-1}$, then for all $Z\in \mathcal{Z}_\Sigma$ initiated from $z$ in $\Re^n$, we have $V_\Sigma(0,Z)=\Psi_{\mathbb{T}}^\circ(z)\geq q(\lVert z\rVert)$ for some $\mathcal{K}_\infty$ function $q$. 
It follows that 
\begin{equation}\label{eq:expl_pr}
\nu_r(M)\leq \sup_{\substack{z\in \Re^n\\ z\neq 0}}\sup_{\substack{\mathbb{I}\subset \mathbb{T}\\|\mathbb{I}|=r}} \dfrac{\Psi_{\mathbb{I}}^\circ(z)}{q(\lVert z\rVert)}
\end{equation}
Hence the more observable (i.e., the larger the gain function $q$), the smaller $\nu_r(M)$. 

For all $(Y,z_0)\in \Re^{n_y\times T}\times \Re^n$ with $Y$ expressed columnwise in the form $Y=\begin{pmatrix}y_0 & \cdots & y_{T-1}\end{pmatrix}$, consider now the following notations:
\begin{align*}
\mathcal{I}^0(Y,z_0)&=\{t\in \mathbb{T}\: :\: y_t-M_tz_0=0\}\\
\mathcal{I}^c(Y,z_0)&=\{t\in \mathbb{T}\: : \:  y_t-M_tz_0\neq0\}. 
\end{align*}

\begin{theorem}[Exact Recoverability Condition]\label{th:ex_rec}
Consider the cost function \eqref{eq:V-Sigma-circ} where $M=\left\{M_t\right\}$ is assumed to have been constructed  as in \eqref{eq:Mt} from the matrices of system \eqref{eq:sys}. Assume that the loss functions   $\{\psi_t\}$ involved in \eqref{eq:V-Sigma-circ} are  defined by~\eqref{eq:psi_f} in which $\psi$ is assumed to satisfy \ref{prop:prem}, \ref{prop:3} and \ref{prop:der} with constant $\gamma_\psi=1$.
Let $r$ be a positive integer.  
If the system \eqref{eq:sys} is observable on $\interval{0}{T-1}$, then the two following propositions are equivalent:
\begin{enumerate}
\item[\emph{(i)}] \label{th_er_i} For all $Y$ in $\Re^{n_y\times T}$ and all $z_0$ in $\Re^{n}$,
\begin{equation}\label{eq:equiv-exact-recovery}
|\mathcal{I}^c(Y,z_0)|\leq r \: \: \Rightarrow \:  \:   \argmin_{z\in \Re^n}V_{\Sigma}^\circ(Y,z)=\big\{z_0\big\}
\end{equation}
\item[\emph{(ii)}] \label{th_er_ii} The index $\nu_r(M)$ satisfies
\begin{equation}\label{eq:nuR<1/2}
\nu_r(M)< 1/2
\end{equation}
\end{enumerate}
\end{theorem}
\begin{proof} 
(i) $\Rightarrow$ (ii): Assume that (i) holds. Consider an arbitrary subset $\mathbb{I}$ of $\mathbb{T}$ such that $|\mathbb{I}|\leq r$. Let $z_0\neq 0$ be a vector in $\Re^n$.  Construct a sequence $Y$ in $\Re^{n_y\times T}$ such that
$y_t=0$ if  $t\in \mathbb{I}$ and $y_t=M_tz_0$ otherwise. Then  $\mathcal{I}^c(Y,z_0)\subset \mathbb{I}$, so that $|\mathcal{I}^c(Y,z_0)|\leq r$. It then follows from (i) that  $\argmin_{z\in \Re^n}V_{\Sigma}^\circ(Y,z)=\big\{z_0\big\}$ which means that $V_{\Sigma}^\circ(Y,z_0)<V_{\Sigma}^\circ(Y,z)$ for all $z\in \Re^n$, $z\neq z_0$. In particular, $V_{\Sigma}^\circ(Y,z_0)<V_{\Sigma}^\circ(Y,0)$ which, by taking into account the definition of $Y$, gives
$$\Psi_{\mathbb{I}}^\circ(z_0)< \Psi_{\mathbb{I}^c}^\circ(z_0),$$
where $\mathbb{I}^c=\mathbb{T}\setminus \mathbb{I}$. Since $\Psi_{\mathbb{T}}^\circ(z_0)=\Psi_{\mathbb{I}}^\circ(z_0)+ \Psi_{\mathbb{I}^c}^\circ(z_0)$, we see that
$$\dfrac{\Psi_{\mathbb{I}}^\circ(z_0)}{\Psi_{\mathbb{T}}^\circ(z_0)}<1/2 $$
This reasoning works for every nonzero $z_0$ and for every subset $\mathbb{I}$ of $\mathbb{T}$. We can hence conclude that $\nu_r(M)<1/2$. 

(ii) $\Rightarrow$ (i): Assume that (ii) holds. Let $(Y,z_0)\in \Re^{n_y\times T}\times \Re^n$ be  such that
$|\mathcal{I}^c(Y,z_0)|\leq r$. We then need to prove that $\argmin_{z\in \Re^n}V_{\Sigma}^\circ(Y,z)=\big\{z_0\big\}$. 
Since the assertion (ii) is assumed true, it follows from \eqref{eq:nuR-M} and \eqref{eq:nuR<1/2}  that 
\begin{equation}\label{eq:INEQ-Psi}
	2\Psi_{\mathbb{I}^c}^\circ(z_0')< \Psi_{\mathbb{T}}^\circ(z_0') \quad \forall z_0'\in \Re^n, \: z_0'\neq 0
\end{equation}
where, for simplicity, we have posed $\mathbb{I}^c=\mathcal{I}^c(Y,z_0)$. In the derivation of \eqref{eq:INEQ-Psi}, we have used the obvious fact that $r_1\leq r_2$ $\Rightarrow$ $\nu_{r_1}(M)\leq \nu_{r_2}(M)$.  If we pose  $\mathbb{I}=\mathcal{I}^0(Y,z_0)=\mathbb{T}\setminus \mathbb{I}^c$, then  the inequality \eqref{eq:INEQ-Psi} is equivalent to
\begin{equation}\label{eq:INEQ-Mtz0}
	\sum_{t\in \mathbb{I}^c}\psi_t(M_tz_0') < \sum_{t\in \mathbb{I}}\psi_t(M_tz_0') 
\end{equation}
Now we observe that for all $t$ in $\mathbb{I}=\mathcal{I}^0(Y,z_0)$, $y_t=M_tz_0$, so that $\psi_t(M_tz_0')=\psi_t\big(y_t-M_t(z_0+z_0')\big)$. On the other hand, for  $t\in\mathbb{I}^c=\mathcal{I}^c(Y,z_0)$, if we apply the GTI~\eqref{eq:tri_gen} with $\gamma_\psi=1$, we obtain 
\begin{align*}
\psi_t(M_tz_0')&=\psi_t\big(y_t-M_tz_0 -(y_t-M_t(z_0+z_0')\big)\\
&\geq \psi_t(y_t-M_tz_0)-\psi_t(y_t-M_t(z_0+z_0'))
\end{align*}
 Combining all these remarks with \eqref{eq:INEQ-Mtz0} yields 
$$
\begin{aligned}
\sum_{t\in \mathbb{I}^c} &\left[\psi_t(y_t-M_tz_0)-\psi_t(y_t-M_t(z_0+z_0')\right]\\
& \hspace{4cm}< \sum_{t\in\mathbb{I}}\psi_t(y_t-M_t(z_0+z_0'))
\end{aligned}
$$
Rearranging this gives $V_{\Sigma}^\circ(Y,z_0)<V_{\Sigma}^\circ(Y,z_0+z_0')$ for all $z_0'\in \Re^n$ with $z_0'\neq z_0$. 
This is equivalent to  $\argmin_{z\in \Re^n}V_{\Sigma}^\circ(Y,z)=\big\{z_0\big\}$. Hence (ii) holds as claimed. 
\end{proof}

From the statement of Theorem \ref{th:ex_rec}, we infer that under the assumption that only the sparse noise $\left\{s_t\right\}$ is active (i.e., there is no dense noise $(w_t,v_t)$) in the system  equations  \eqref{eq:sys},  $\mathcal{E}^\circ(Y)=\left\{X\right\}$ whenever $\nu_r(M)<1/2$, i.e, the estimator $\mathcal{E}^\circ$ returns exactly the true state.  
For a given system, if one can evaluate numerically the index $\nu_r(M)$, then it becomes possible to assess the number $r_{\max}\triangleq\max\left\{r:\nu_r(M)<1/2\right\}$ of gross errors that can be corrected by the estimator $\mathcal{E}^\circ$ when applied to that system. We will get back to the computational matter in Section \ref{sec:Numerical-Eval}.

\subsection{Special case of $\ell_0$-norm loss based estimator}
Consider the special case where the loss functions $\left\{\psi_t\right\}$ are defined, for all $t\in \mathbb{T}$, by 
\begin{equation}\label{eq:L0-Norm}
	\forall e\in \Re^{n_y}, \quad \psi_t(e)= \left\{\begin{array}{ll} 1 &  \mbox{if }\:  e\neq 0\\ 0 & \mbox{otherwise}\end{array} \right. 
\end{equation}
This corresponds to the block $\ell_0$-norm. Note that such functions satisfy the assumptions \ref{prop:prem}, \ref{prop:3} and \ref{prop:der} requested in the definition \ref{def:concentration-ratio} of $\nu_r(M)$ and in the statement of Theorem \ref{th:ex_rec}. Hence $\nu_r(M)$ is well-defined in this case. 
\begin{coro}\label{coro:L0-recovery}
Consider system \eqref{eq:sys} under the assumption that $w_t=0$ for all $t$. Assume observability of the system on $\interval{0}{T-1}$. Consider the estimator \eqref{eq:Constrained-Estimator} in which the cost $V_{\Sigma}^\circ$ is defined from the family of loss functions $\left\{\psi_t\right\}$ expressed in \eqref{eq:L0-Norm}. Then for all $(Y,z_0)\in \Re^{n_y\times T}\times \Re^n$, 
$$|\mathcal{I}^c(Y,z_0)|<\dfrac{T-\mu(M)+1}{2} \: \Rightarrow \: \mathcal{E}^\circ(Y)=\left\{X\right\},  $$
where $\mu(M)$ defined by 
\begin{equation}\label{eq:nu(M)}
	\begin{aligned}
		\mu(M)=&\min\Big\{k:\forall I\subset \mathbb{T}, \Big(|I|=k\: \Rightarrow \: \rank(M_I)=n\Big) \Big\} 
	\end{aligned}
\end{equation}
is the minimum number $k$ such that any matrix $M_I\in \Re^{|I|n_y\times n}$ formed by stacking vertically the matrices of the collection $\left\{M_t:t\in I\right\}$ indexed by  $I\subset \mathbb{T}$ with $|I|=k$, has full column rank. 
\end{coro}
\begin{proof}
Let us start by observing that with the particular loss functions invoked in the statement of the corollary, 
 $\Psi_{\mathbb{T}}^\circ(z)$ denotes the number of $t\in \mathbb{T}$ for which $\psi_t(M_tz)\neq 0$. It follows from the definition of $\mu(M)$ that  $\Psi_{\mathbb{T}}^\circ(z)\geq T-\mu(M)+1$ for all $z\neq 0$.  The reason for this is that if $\psi_t(M_tz)$ was to be equal to zero more than $\mu(M)-1$ times, then $z$ would be necessarily equal to zero. 
As a result we get 
$$\nu_r(M)\leq \dfrac{r}{T-\mu(M)+1} $$
Hence by applying Theorem \ref{th:ex_rec}, $|\mathcal{I}^c(Y,z_0)|/(T-\mu(M)+1)<1/2$ is a sufficient condition for exact recovery by the $\ell_0$-norm based estimator. 
\end{proof}

\begin{rem}\label{rem:EntryWise-L0-Estimator}
Assume that $\psi_t$ is defined to be the counting norm, i.e.,  
\begin{equation}\label{eq:L0-Norm-EntryWise}
	\psi_t(e)= \left\|e\right\|_0
\end{equation}
Then $\psi_t$ has a separable structure as illustrated in \eqref{eq:psi_dec}. Consider then defining, still under the observability assumption, an entrywise version of the concentration ratio by
\begin{equation}\label{eq:nuR-M-Tilde}
\tilde{\nu}_r(M)=\sup_{\substack{z\in \Re^n\\ z\neq 0}}\sup_{\substack{\mathbb{I}\subset \mathbb{T}\times \mathbb{S}\\|\mathbb{I}|=r}}\dfrac{\sum_{(t,i)\in \mathbb{I}}\left\|M_{ti}z\right\|_0}{\Psi_{\mathbb{T}}^\circ(z)}
\end{equation}
where $M_{ti}$ refers to the $i$-th row of $M_t$. 
Further, let
$$\tilde{\mu}(M)=\min\Big\{k:\forall I\subset \mathbb{T}\times \mathbb{S}, \Big(|I|=k\: \Rightarrow \: \rank(\tilde{M}_I)=n\Big)\Big\} $$
with $\tilde{M}_I\in \Re^{|I|\times n}$ denoting the matrix obtained by stacking the row vectors $\left\{M_{ti}: (t,i) \in I\right\}$. 
Then a result similar to Corollary \ref{coro:L0-recovery} is obtainable: if the number  the measurements corrupted by a nonzero error (among the $n_yT$ available) is strictly less than $(n_yT-\tilde{\mu}(M)+1)/2$, then the estimator $\mathcal{E}^\circ$  expressed in \eqref{eq:Constrained-Estimator} (with $\psi_t$ being the $\ell_0$ norm as in \eqref{eq:L0-Norm-EntryWise}) recovers exactly the true state.  
\end{rem}
\begin{rem}
Under the condition of Remark \ref{rem:EntryWise-L0-Estimator}, if we consider the scenario where only a set of  $k<n_y$  sensors may be compromised by attackers, then  exact recovery is achieved if 
\begin{equation}\label{eq:cond-number-sensors}
	k<\frac{n_y}{2}-\dfrac{\tilde{\mu}(M)-1}{2T}.
\end{equation}
Taking into consideration the fact that  $\tilde{\mu}(M)-1<T$, it can then be seen that \eqref{eq:cond-number-sensors} is equivalent to  $k\leq \left\lceil n_y/2-1\right\rceil$ where the notation $\left\lceil r\right\rceil$, for $r\in \Re$, refers to the smallest integer larger or equal to $r$. 
To sum up, when the $\psi_t$ are defined as in \eqref{eq:L0-Norm-EntryWise},  the estimator \eqref{eq:Constrained-Estimator} is able to return the true state matrix even when $\left\lceil n_y/2-1\right\rceil$  sensors get faulty over the entire observation horizon.  This is reminiscent of a result stated in \cite{fawzi_secure_2014} which therefore appears to be a consequence of Theorem \ref{th:ex_rec}. 
\end{rem}

\subsection{Stability of the class of  estimators $\mathcal{E}^\circ$ with respect to dense noise}
We have argued that  the class of estimators $\mathcal{E}^\circ$ in  \eqref{eq:Constrained-Estimator} is able to obtain exactly the true state matrix when there is no dense noises $(w_t,v_t)$ in the system equations and only the sparse noise $\left\{s_t\right\}$ is active.  The question we ask now is whether this set of estimators can, in addition to sparse noise, handle dense process and output noises and  to what extent this is possible. The starting point of our reflection is the observation that the dynamical system defined by 
\begin{equation}\label{eq:New-sys}
	\begin{aligned}
	&	\tilde{x}_{t+1}=A_t \tilde{x}_t, \quad \tilde{x}_0 = x_0 \\
	&	y_t = C_t\tilde{x}_t+s_t+(v_t+\tilde{v}_t),
	\end{aligned}
\end{equation}
produces the same output as system \eqref{eq:sys}. Here, $\tilde{v}_t=C_t\tilde{w}_t$, with $\tilde{w}_t=\sum_{k=0}^{t-1}A_{t-1}\cdots A_{k+1}w_k$ a definition which uses the convention that the product $A_{t-1}\cdots A_{k+1}=I$ if $k= t-1$. Then the idea is to apply the estimator $\mathcal{E}^\circ$ to \eqref{eq:New-sys}  by neglecting the dense component $(v_t+\tilde{v}_t)$ of the output equation. To state the resilience result for $\mathcal{E}^\circ$, consider for a given $\varepsilon\geq 0$, a partition $(\tilde{\mathbb{T}}_\varepsilon,\tilde{\mathbb{T}}_\varepsilon^c)$ of $\mathbb{T}$ defined as in \eqref{eq:def_teps} with $f_t$ replaced by $\tilde{f}_t\triangleq s_t+(v_t+\tilde{v}_t)=f_t+\tilde{v}_t$. 
\begin{theorem}\label{thm:Stability-E0}
Consider the estimator \eqref{eq:Constrained-Estimator} for the system \eqref{eq:sys}. 
Assume that the loss functions   $\{\psi_t\}$ involved in \eqref{eq:V-Sigma-circ} are  defined by~\eqref{eq:psi_f} in which $\psi$ is assumed to satisfy \ref{prop:prem}--\ref{prop:der} with constant $\gamma_\psi=1$.
Let $\varepsilon\geq 0$ and set $r=|\tilde{\mathbb{T}}_\varepsilon^c|$. 
Denote with $N$ a norm on $\Re^{n\times T}$ defined by $N(Z)=\max_{t\in \mathbb{T}}\left\|z_t\right\|$ with $z_t$ being the $t$-th column of $Z$ and $\left\|\cdot\right\|$ being a norm on $\Re^n$.   \\ 
If the system \eqref{eq:sys} is observable on $\interval{0}{T-1}$ and $\nu_r(M)<1/2$,  then  there exists a $\mathcal{K}_\infty$ function $\alpha$ such that for all norm  $\lVert \cdot \rVert$ on $\Re^{n\times T}$,  
\begin{equation}
	N(\hat{X}-X) \leq R_\Sigma\alpha^{-1}(\rho) +\max_{t\in \mathbb{T}}\left\|\tilde{w}_t\right\|\quad \forall  \hat{X}\in \mathcal{E}^\circ(Y), 
\end{equation}
where $R_\Sigma$ is some  constant depending on the system $\Sigma$ and 
$$\rho=\dfrac{2}{D_1\big(1-2\nu_r(M)\big)}\sum_{t\in \tilde{\mathbb{T}}_\varepsilon}\psi_t(\tilde{f}_t)$$ 
with $\tilde{f}_t=f_t+\tilde{v}_t$, and $D_1=\min_{\left\|z\right\|=1} \Psi_{\mathbb{T}}^\circ(z)$. 
\end{theorem}
\begin{proof}
Let $\hat{x}_0\in \argmin_{z\in \Re^n}V_{\Sigma}^\circ(Y,z)$. We first provide a bound on the error 
$e_0=\hat{x}_0-x_0$ with $x_0$ denoting the true initial state of system \eqref{eq:sys}. By exploiting the fact that $V_{\Sigma}^\circ(Y,\hat{x_0})\leq V_{\Sigma}^\circ(Y,x_0)$ and noting that $y_t=M_tx_0+\tilde{f}_t$, we reach the inequality
$$\sum_{t\in \mathbb{T}}\psi_t(\tilde{f}_t-M_te_0)\leq  \sum_{t\in \mathbb{T}}\psi_t(\tilde{f}_t).  $$
By then reasoning quite similarly as in the proof of Theorem \ref{th:lb}, we get
$$\Psi_{\mathbb{T}}^\circ(e_0)-2\Psi_{\mathbb{T}_\varepsilon^c}^\circ(e_0)\leq \sum_{t\in \tilde{\mathbb{T}}_\varepsilon}\psi_t(\tilde{f}_t)$$
which, by exploiting \eqref{eq:nuR-M} and the assumption that $\nu_r(M)<1/2$, leads to
$$\Psi_{\mathbb{T}}^\circ(e_0)\leq \dfrac{2}{1-2\nu_r(M)} \sum_{t\in \tilde{\mathbb{T}}_\varepsilon}\psi_t(\tilde{f}_t) $$
Applying  now Lemma~\ref{lem:lb}, we conclude that for any norm $\lVert \cdot \rVert$ on $\Re^n$, there exists a $\mathcal{K}_\infty$ function $\alpha$ such that
$\left\|e_0\right\|\leq \alpha^{-1}(\rho)$. 
Now by observing that for any $\hat{X}\in \mathcal{E}^\circ(Y)$, 
$$
\begin{aligned}
	\hat{X}-X =&\begin{pmatrix} e_0 & A_0 e_0 & \cdots & A_{T-1}\cdots A_0 e_0 \end{pmatrix}\\
	&\hspace{3cm} -\begin{pmatrix} 0 & \tilde{w}_0 & \cdots & \tilde{w}_{T-1}\end{pmatrix}
\end{aligned}
$$
the result follows by posing\footnote{We use here the convention that $A_{t-1}\cdots A_0=I$ if $t=0$. } 
$R_\Sigma=\max_{t\in \mathbb{T}}\left\|A_{t-1}\cdots A_0\right\|_{\text{ind}}$ with $\left\|\cdot\right\|_{\text{ind}}$ being the matrix norm induced by the vector norm $\left\|\cdot\right\|$  on $\Re^n$.
\end{proof}
\noindent The interest in Theorem \ref{thm:Stability-E0}   is that it provides a condition of resilience for the estimator $\mathcal{E}^\circ$ which can be checked numerically as will  be  discussed in the next section. 
\section{On the numerical evaluation of the resilience conditions}\label{sec:Numerical-Eval}
The analysis results presented in Sections \ref{part:iv} and \ref{sec:Exact-Recoverability} rely on some functions (resilience index, concentration ratio, \ldots) which characterize quantitatively some properties of the system being observed. A question we ask now is whether it would be possible to evaluate numerically these measures.  In effect, computing the $r$-resilience index in  \eqref{eq:def_pr} would help testing for example the resilience condition in Theorem \ref{th:lb}. Similarly, evaluating the concentration ratio $\nu_r(M)$ introduced in \eqref{eq:nuR-M} is the way to assess whether a given estimator is able to return the true state of a given system if we make an hypothesis on the number of potential nonzero errors in the measurements.

Unfortunately, obtaining numerically the numbers $p_r$, $\tilde{p}_r$ or $\nu_r$ require solving some hard nonconvex and combinatorial optimization problems. This is indeed a common characteristic of the concepts which are usually used to assess resilience; for example, the popular  Restricted Isometry Property (RIP) constant ~\cite{candes_restricted_2008} is comparatively as hard to evaluate. 
We note however that when the dimension of the state is small enough, $\nu_r(M)$ can be exactly computed  by taking inspiration from a method presented in \cite{sharon_minimum_2009} even though at the price of a huge (but affordable) computational cost. Alternatively, a cheaper overestimation can be obtained by means of convex optimization as suggested in \cite{bako_class_2017}. 
The next lemma provides such an  overestimate for $\nu_r(M)$. 
\begin{lem}[An estimate of $\nu_r$]\label{lem:OverEstimate-nuR}
Assuming all quantities are well-defined (see the conditions in Definitions \ref{def:pr} and \ref{def:concentration-ratio}), the following statements hold: 
\begin{itemize}
	\item[(a)] $\nu_r\leq p_r$ 
\item[(b)] If $\mu(M)\leq T-1$ then 
\begin{equation}\label{eq:Upper-Bound-nuR}
	\nu_r(M)\leq \dfrac{r\nu^o}{1+\nu^o},
\end{equation}
where 
\begin{equation}\label{eq:nu0}
	\nu^o=\max_{t\in \mathbb{T}}\min_{\lambda_t\in \Re^T}\Big\{\left\|\lambda_t\right\|_\infty: V_tM_t=\sum_{k\in \mathbb{T}}\lambda_{tk}V_kM_k, \lambda_{tt}=0\Big\}
\end{equation}
\end{itemize}
In \eqref{eq:nu0}, the $\lambda_{tk}$ denote the entries of the vector $\lambda_t\in \Re^T$ and $\left\{V_t\right\}$ refers to the sequence of nonsingular weighting matrices involved in  \eqref{eq:psi_f}.  
\end{lem}
\noindent The proof of statement (a) is straightforward by noticing that \eqref{eq:nuR-M} follows from \eqref{eq:def_pr} by constraining the variable $Z$ to be in $\mathcal{Z}_\Sigma$. As to the proof of statement (b), it follows a similar reasoning as the proof of Theorem 2 in \cite{bako_class_2017}. \\
The interest of this lemma is twofold. First it suggests that the resilience condition of $\mathcal{E}^\circ$ is weaker (in the sense that it is easier to achieve) than that of $\mathcal{E}$. Second, it provides an upper bound on $\nu_r(M)$ which can be computed by solving a convex optimization problem (see Eqs \eqref{eq:Upper-Bound-nuR}-\eqref{eq:nu0}). More specifically, given $\nu^o$ in \eqref{eq:nu0}, checking numerically whether $|\mathbb{T}_\varepsilon^c|<1/2(1+1/\nu^o)$ provides a sufficient condition for $\nu_r(M)<1/2$ and so, for the resilience of the estimator \eqref{eq:Constrained-Estimator}.   
In a similar spirit as in Lemma \ref{lem:OverEstimate-nuR}, we now show that the parameter $\tilde{p}_r$ defined in \eqref{eq:pr_norm1} can also be overestimated via convex optimization if the loss functions $\phi_t$ and $\psi_t$ in~\eqref{eq:phi_f}--\eqref{eq:psi_f} are both taken to be norms. 
\begin{lem}[An estimate of $\tilde{p}_r$]\label{lem:pr_calc}
Consider the resilience parameter $\tilde{p}_r$ defined in \eqref{eq:pr_norm1} where we assume that   $\psi_{ti}(e)=|e|$ for all $(t,i,e)\in \mathbb{T}\times \mathbb{S}\times \Re$ and $\phi_t$ is an arbitrary  norm. Then
\begin{equation}\label{eq:pr_overestimate}
	\tilde{p}_r\leq \dfrac{r}{b_1}
\end{equation}
where
\begin{equation}\label{eq:def_b1}
b_1=\inf_{(t,i )\in\mathbb{T}\times\mathbb{S}}\:\inf_{\substack{Z\in \Re^{n\times T}}}\left\{H_\Sigma(Z): c_{ti}^\top z_\tau = 1\right\}=\dfrac{1}{\tilde{p}_1}.
\end{equation}
\end{lem}

\begin{proof}
(See Appendix~\ref{app:pr_proof})
\end{proof}
\noindent Note, under the assumptions of Lemma \ref{lem:pr_calc}, that $\inf_{\substack{Z\in \Re^{n\times T}}}\left\{H_\Sigma(Z): c_{ti}^\top z_\tau = 1\right\}$ is a convex optimization problem for any given $(t,i)$. Hence, solving for $b_1$ in \eqref{eq:def_b1}  requires solving $Tn_y$ convex problems and picking the smallest value among all. The interest of the lemma is that it provides an overestimate of $\tilde{p}_r$ which is numerically computable. 
Based on the so obtained overestimate of $\tilde{p}_r$, we see from Theorem~\ref{th:lb} that the estimator~\eqref{eq:def_est} is resilient to $r$ outliers if $r<b_1/2$. Moreover,   we can deduce an underestimate of the number of outliers that the estimator~\eqref{eq:def_est} is able to handle  as $r_{\max}=\max\left\{r: r<b_1/2\right\}$.

\noindent As a last remark in this paragraph, let us observe that   Lemma \ref{lem:pr_calc} is also applicable to overestimate $p_r$ defined in~\eqref{eq:def_pr} in the case of a single-output system, \textit{i.e.}, when  $n_y=1$.

\color{black}

\section{Simulation Results}\label{part:v}

In this part we will illustrate numerically the resilience properties of the proposed class of estimators. 
For this purpose, we consider for simplicity, an example of linear time-invariant system in the form \eqref{eq:sys}. 
We select a single-input single-output example where the pair $(A,C)$ is given by 
\begin{equation}\label{eq:example-matrices}
	\begin{aligned}
		&A = \begin{pmatrix} 0.7 & 0.45\\
   -0.5 & 1\end{pmatrix}, \quad 
		C = \begin{pmatrix}1 & 2\end{pmatrix}. 
	\end{aligned}
\end{equation}
We instantiate the loss functions in \eqref{eq:V} as follows:   For all $t$ in $\mathbb{T}$ and for all $(z,e)\in \Re^{n}\times \Re^{n_y}$, $\phi_t(z)=\psi(W_t z)$ and  $\psi_t(e)=\psi(V_te)$ where the weighting matrices $W_t$ and $V_t$ and the  functions $\phi$ and $\psi$ will be specified below  for each experiment.

\subsection{Numerical certificate of exact recoverability}\label{subsec:Exact-Recoverability}
Suppose in this section that the process noise $w_t$  and the dense component $v_t$ of $f_t$ (see Eq. \eqref{eq:f=s+v}) are both identically equal to zero.  
We then focus on testing the exact recoverability property of the estimator \eqref{eq:Constrained-Estimator} in the presence only of the sparse noise $\left\{s_t\right\}$.  The times of occurrence of the nonzeros values in the sequence $\left\{s_t\right\}$ are picked at random. As to its values there are also randomly generated from  a zero-mean normal distribution with variance $100^2$. Given $T=100$ output measurements and the system matrices in \eqref{eq:example-matrices}, the estimator $\mathcal{E}^\circ$ is implemented by directly solving the optimization problem defined in \eqref{eq:Constrained-Estimator} through the CVX interface \cite{grant_cvx_2017}. Note that the implementation of the estimator \eqref{eq:Constrained-Estimator}-\eqref{eq:V-Sigma-circ} requires computing the matrices $M_t$ expressed in \eqref{eq:Mt}, which take the form $CA^t$ in the LTI case. A problem that may occur however  is that if $A$ is Schur stable as is the case here (or unstable), taking successive powers of $A$ produces matrices $M_t$ which might not be of the same order of magnitude. To preserve the contribution of each term of \eqref{eq:V-Sigma-circ},  we introduce special weighting matrices $\left\{V_t\right\}$ (in the loss function $\psi$ selected as the $\ell_1$ norm) to normalize the rows  of these matrices so that they all have unit $2$-norm. 
$V_t$ is therefore selected to be a diagonal matrix of the form $V_t=\diag(V_{t1}, \cdots, V_{tn_y})$, where 
\begin{equation}\label{eq:Vti}
	V_{ti}=
	\left\{\begin{array}{ll}1/\left\|c_i^\top A^t\right\|_2 & \mbox{if } \: c_i^\top A^t\neq 0
	\\  1  & \mbox{otherwise} \end{array}. \right. 
\end{equation}
Here,  $c_i^\top$, $i=1,\ldots,n_y$,  denote the $i$-th row of the matrix $C$. Indeed the effect of the weighting function in \eqref{eq:V-Sigma-circ} is equivalent to changing $y_t$ and $M_t$ respectively to $\tilde{y}_t=V_ty_t$  and $\tilde{M}_t=V_tCA^t$. 
Posing  $M=\big\{\tilde{M}_t\big\}$, it can be checked using the methods discussed in Section \ref{sec:Numerical-Eval} (See Eq. \eqref{eq:Upper-Bound-nuR}) that at least $r_{\max}=30$  erroneous data (out of $T=100$ measurements)  can be accommodated by the estimator while still returning exactly the true state.

\noindent To investigate empirical performance, we consider different ratios $|\Lambda_0^c|/T$ of nonzero values in the sequence $\left\{s_t\right\}$. For each fixed proportion of nonzero values, we run the estimator over $100$ different realizations of the output measurements.  The results, depicted in Figure \ref{fig:exact_recovery_without_normalization}, tend to show that the estimator can still find the true state even for proportions of gross errors as large as $60\%$. 
\begin{figure}
	\centering
	\psfrag{Normalization}{\scriptsize With normalization}
	\psfrag{WithoutNormalization}{\scriptsize Without normalization}
			\includegraphics[scale=.45]{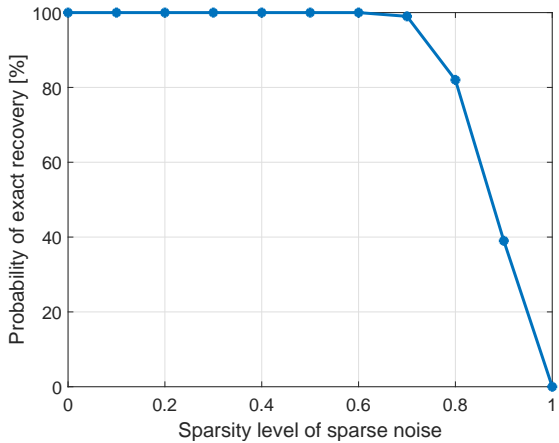} 
		\captionsetup{justification=justified}
	\caption{Probability of exact recovery (expressed in percentage) by the estimator \eqref{eq:Constrained-Estimator} in the presence of only sparse measurement noise $\left\{s_t\right\}$. The level of sparsity of the noise is expressed in terms of a fraction of nonzero values in the sequence $\left\{s_t: t\in \mathbb{T}\right\}$ with $\left|\mathbb{T}\right|=T=100$. 
	}
		\label{fig:exact_recovery_without_normalization}
\end{figure}

\subsection{Performances in the presence of dense noise}
We consider now the more realistic scenario where the process noise $\left\{w_t\right\}$ and the measurement noise $\left\{v_t\right\}$ are nonzero. We further assume them to be bounded, white  and uniformly distributed. For the numerical experiments these signals are sampled from   an interval of the form $\interval{-a}{a}$.
For comparison purpose, we conduct the estimation with several  estimators:
\begin{itemize}
	\item an instance of the estimator \eqref{eq:def_est},  denoted  $\mathcal{E}_{\ell_2^2,\ell_1}$ in the sequel, in which the loss function $\phi_t$ is quadratic and  $\psi_t$ is  the $\ell_1$-norm and $\lambda=1000$ (see \eqref{eq:V} and \eqref{eq:phi_f}-\eqref{eq:psi_f}) 
	\item an instance of the estimator \eqref{eq:def_est} denoted  $\mathcal{E}_{\ell_1,\ell_1}$  in which both loss functions $\phi_t$ and  $\psi_t$ are  the $\ell_1$-norm with $\lambda=10$
	\item the estimator $\mathcal{E}^\circ$ defined in \eqref{eq:Constrained-Estimator}
\end{itemize}
In addition we implement  oracle versions\footnote{By oracle version of an estimator, we refer here to an implementation of this estimator which is aware of the sparse noise sequence $\left\{s_t\right\}$. } of $\mathcal{E}_{\ell_2^2,\ell_1}$ and of $\mathcal{E}_{\ell_2^2,\ell_2^2}$ (the latter corresponding to an instance of $\mathcal{E}$ where both $\phi_t$ and $\psi_t$ are instantiated as quadratic functions).

\paragraph{Experiment 1: Resilience test} Keeping the level of both dense noises (i.e., $w_t$ and $v_t$) fixed with amplitude $a=0.03$ for the entries of the former and $a=0.1$ for the latter (yielding a Signal to Noise Ratio (SNR) of about 30 dB in each case), we apply the estimators  $\mathcal{E}_{\ell_2^2,\ell_1}$, $\mathcal{E}_{\ell_1,\ell_1}$ and $\mathcal{E}^\circ$ as defined above) 
to  $100$ different realizations of the output data and we compute  the average of the corresponding  relative estimation errors. This process is repeated for different fractions of nonzeros in the sparse noise $\left\{s_t\right\}$ ranging from $0$ to $0.8$. The estimates obtained by these estimators are displayed in Figure \ref{fig:Noisy-Estimate} in log scale. For the sake of comparison, we also display the oracle estimates given by $\mathcal{E}_{\ell_2^2,\ell_1}$ and those obtained by a standard least squares estimator $\mathcal{E}_{\ell_2^2,\ell_2^2}$ (i.e. with $\phi_t$ and $\psi_t$ taken to be both quadratic in \eqref{eq:V}). By oracle of an estimator, we mean here a version of that estimator which is aware of the true values of the sparse noise sequence $\left\{s_t\right\}$. 
The results tend to show that the estimator \eqref{eq:def_est} remains stable until the (empirical) resilience condition is violated (an event that happens when the sparsity level for the sparse noise is around $60\%$). This is consistent with the resilience property characterized in Theorem \ref{th:lb} and the empirical observations made in Section \ref{subsec:Exact-Recoverability} according to which the estimator is insensitive to the sparse noise sequence $\left\{s_t\right\}$ as long as the number of nonzero values in it (whose magnitudes are possibly arbitrarily large) is less than a certain threshold determined by the properties of the system. 
 While Lemma \ref{lem:pr_calc} provides an underestimate of the number of correctable outliers as $r_{\max}=8$ (out of $100$), we can observe that the empirical breakpoint in the current example seems to be indeed around $40\%$. The discrepancy between the two values is partly explained by the pessimism of the upper bound of $p_r$ proposed in Lemma  \ref{lem:pr_calc}. 
%

\begin{figure}[h!]
	\centering
	\includegraphics{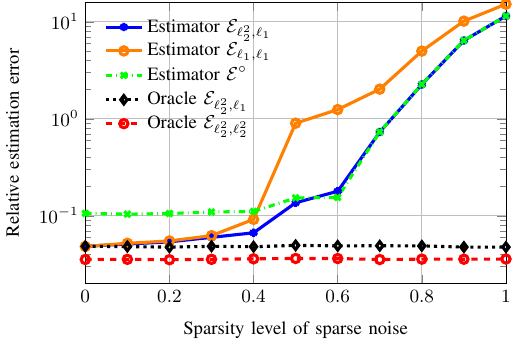}
		\captionsetup{justification=justified}
	\caption{Average relative estimation error (in logarithm scale) induced by  different estimators  versus sparsity level of the sparse noise $\left\{s_t\right\}$. The relative error is expressed here as $\big\|\hat{X}-X\big\|_2/\left\|X\right\|_2$ where $X$ and $\hat{X}$ denote the true and estimated state matrices respectively. Parameters of the estimator $\mathcal{E}$ in \eqref{eq:def_est}: $\lambda=1000$, $W_t=I_2$ and $V_t=1$ for all $t$.  
	}
		\label{fig:Noisy-Estimate}
\end{figure}

\paragraph{Experiment 2: Stability with respect to dense noise} Now, we fix the sparsity level of the time sequence $\left\{s_t\right\}$ to $0.2$ and let the powers of the dense noise $\left\{(w_t,v_t)\right\}$ vary jointly from 5 dB to 100 dB in term of SNR. 
The estimates obtained by the estimators \eqref{eq:def_est} and \eqref{eq:Constrained-Estimator} with the choices of $\phi_t$ and $\psi_t$ agreed in the beginning of Section \ref{part:v} are displayed in Figure \ref{fig:Noisy-Estimate-Different-Levels-Noise} in term of $\log_{10}$ of estimation errors. What this illustrates is that whenever the number of faulty data is reasonable (here $20\%$ of the available measurements), the  estimator discussed in this section behaves almost in the same way as when there is no faulty data at all.    
\begin{figure}[h!]
	\centering
	\includegraphics{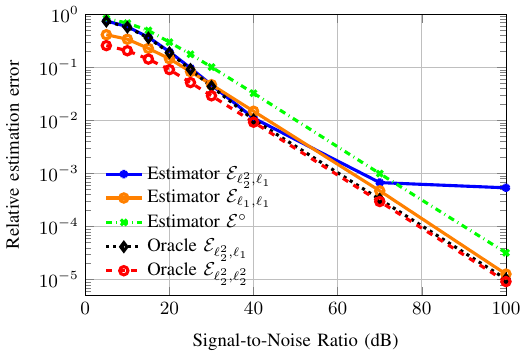}
		\captionsetup{justification=justified}
	\caption{Average relative estimation error (in log scale) induced by different estimators  for different levels of both dense noises $w_t$ and $v_t$.  Parameters of the estimator $\mathcal{E}$ in \eqref{eq:def_est}: $\lambda=1000$, $W_t=I_2$ and $V_t=1$ for all $t$.  
	}
		\label{fig:Noisy-Estimate-Different-Levels-Noise}
\end{figure}

\paragraph{Experiment 3: Impact of the regularization parameter $\lambda$ in $\mathcal{E}$} 
To assess the influence of the regularization parameter $\lambda$ on the performance of the estimators $\mathcal{E}_{\ell_2^2,\ell_1}$ and $\mathcal{E}_{\ell_1,\ell_1}$, we fix the amplitudes of both dense noises $\left\{w_t\right\}$ and $\left\{v_t\right\}$ at the same level as in Experiment 1 (\textit{i.e.} SNR equal to $30$dB and ratio of non-zeros entries in $\{s_t\}$ equal to 30$\%$). In this setting, we consider a set of values of $\lambda$ ranging from $10^{-3}$ to $10^{6}$. For each of these values we perform an estimation over a hundred realizations of the output data and compute the average of the corresponding relative estimation errors. The outcome of this test, depicted in Figure~\ref{fig:perf_lambda},   tends to suggest that low values of $\lambda$ yield quite poor results. Conversely, when $\lambda$ tends towards infinity, both estimators' performance measures  saturate at the same value, namely $0.1$. It turns out that this limit value corresponds to the relative error obtained for $\mathcal{E}^\circ$ in Experiment 1 in the same configuration, hence suggesting that $\mathcal{E}$ tends indeed to $\mathcal{E}^\circ$ in behavior as $\lambda$ becomes large. Finally, it is interesting to observe that both performance curves exhibit minima located around $\lambda=750$ and $\lambda=10$ for $\mathcal{E}_{\ell_2^2,\ell_1} $and $\mathcal{E}_{\ell_1,\ell_1}$ respectively.  
\begin{figure}[hbtp]
\centering
\includegraphics{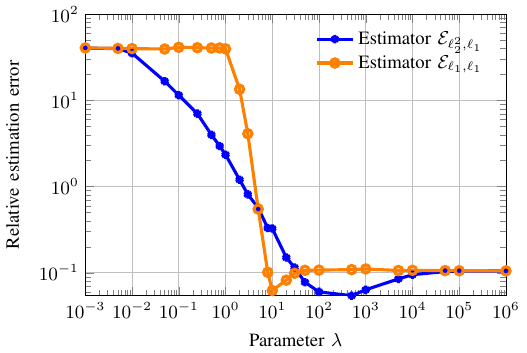}
\caption{Average relative estimation error induced by $\mathcal{E}_{\ell_2^2,\ell_1}$ and $\mathcal{E}_{\ell_1,\ell_1}$ for the system \eqref{eq:example-matrices}, SNR$=30$dB and 30$\%$ of non-zero entries in $\{s_t\}$ for different values of the regularization parameter $\lambda$.}
\label{fig:perf_lambda}
\end{figure}

\section{Conclusion}\label{part:vi}
In this paper, we have considered the problem of estimating the state of linear time-varying systems in the face of uncertainties modeled as process  and measurement noises in the system equations. The measurement noise sequence assumes values of possibly arbitrarily large amplitude which occur intermittently in time and accross the available sensors.   
For this problem we  have proposed a class of estimators based on the resolution of a family of parameterizable optimization problems. The discussed family  is rich enough to include optimization-based estimators based on various loss functions which may be convex (e.g., $\ell_p$-norms) or nonconvex (e.g., $\ell_p$ quasi-norms or saturated functions), smooth or nonsmooth. 
 In particular, we have proved a resilience property for the proposed class of state estimators,  that is, the resulting estimation error is bounded by a bound which is independent of the extreme values of the measurement noise provided that the number of occurrences (over time and over the whole set of sensors) of such extreme values is limited. Note however that the estimators studied here operate in batch mode, that is, they apply to a finite collection of measurements. In future works we intend to investigate efficient and low cost adaptive versions of the proposed optimization framework.   Another interesting research avenue would be to study the level of  performance which is achievable  if one uses the discussed framework as a method to detect bad data prior to a refinement with standard least squares estimation.

\appendix
In this appendix, we provide some technical results  used in the paper and the associated proofs. 

\subsection{A useful technical lemma}\label{app:fc_proof}
\begin{lem}\label{lem:fc}
 Let $\xi_1,\xi_2:\Re^{a\times b}\rightarrow\Re_{\geq 0}$ be two functions which satisfy properties~\ref{prop:prem}--\ref{prop:3} and let $\ell:\Re^{c\times d} \rightarrow\Re^{a\times b}$ be an injective linear mapping. Then $\xi_1+\xi_2$ and $\xi_1\circ\ell$  verify~\ref{prop:prem}--\ref{prop:3}. In addition, the following holds: 
\begin{enumerate}\topsep=0pt\itemsep=0pt
\item[\emph{(j)}] If $\xi_1$, $\xi_2$ verify~\ref{prop:4}, then $\xi_1+\xi_2$ and $\xi_1\circ\ell$ verify \ref{prop:4} \label{lem_cf_i}.
\item[\emph{(jj)}]  If $\xi_1$, $\xi_2$ verify~\ref{prop:der}, then $\xi_1+\xi_2$ and $\xi_1\circ\ell$ verify \ref{prop:der} \label{lem_cf_ii}.
\end{enumerate}
\end{lem}

\noindent The main point of interest of this lemma is that even if there are functions which satisfy properties~\ref{prop:4} and~\ref{prop:der} with different values of $q$ and $\gamma$, their sum still verifies those properties.

\noindent To prove Lemma~\ref{lem:fc}, we will need the following result.  
\begin{lem}[Minimum function of two $\mathcal{K}_{\infty}$ functions]\label{lem:kf}
If $q_1$ and $q_2$ are two $\mathcal{K}_\infty$ functions, then so is the function $q$ defined by 
\begin{equation}
\forall \lambda\in\Re_{\geq 0},\: q(\lambda)=\min_{i\in\{1,2\}}q_i(\lambda)
\end{equation}
\end{lem}
\begin{proof}
We have to prove that $q$ is continuous, strictly increasing and satisfies  $q(0)=0$ and $\lim_{\lambda \rightarrow+\infty}q(\lambda)=+\infty$.\\
First of all, it is clear that $q(0)=0$. Also, continuity of $q$ is immediate from that of $q_1$ and $q_2$ by noting that $q=(q_1+q_2-|q_1-q_2|)/2$. 
To see the strict increasingness of $q$, consider $\lambda_1$ and $\lambda_2$ in $\Re_{\geq 0}$ such that $\lambda_1<\lambda_2$. Then 
$q(\lambda_1)\leq q_1(\lambda_1)<q_1(\lambda_2)$ and  $q(\lambda_1)\leq q_2(\lambda_1)<q_2(\lambda_2)$.  It follows that $q(\lambda_1)<\min_{i\in\{1,2\}}q_i(\lambda_2)=q(\lambda_2)$ and hence $q$ is strictly increasing. 
We now show that $q(\lambda)$ tends to infinity when $\lambda\rightarrow +\infty$. 
Let $M>0$ be an arbitrary positive number. 
 Since $q_1$ and $q_2$ tend to infinity, there exist $\eta_1$ and $\eta_2$ such that
$\lambda\geq \eta_1 \Rightarrow q_1(\lambda)\geq M$ and 
$\lambda\geq \eta_2 \Rightarrow q_2(\lambda)\geq M$. 
By taking $\eta=\max_{i\in\{1,2\}}\eta_i$, it holds that $q(\lambda)\geq M$ whenever $\lambda\geq \eta$, or equivalently that, $\lim_{\lambda\rightarrow +\infty}q(\lambda)=+\infty$. 
\end{proof}

\noindent \textbf{Proof of Lemma \ref{lem:fc}:}
The sum $\xi_1+\xi_2$   has clearly the properties~\ref{prop:prem}--\ref{prop:3} as a sum of continuous, even, positive definite functions. Moreover, the composition of a continuous, even, convex positive definite function with an injective linear mapping yields a continuous, even, positive definite function, so $\xi_1\circ\ell$ satisfies properties~\ref{prop:prem}--\ref{prop:3} too.
 
\noindent \textit{Proof of (j): }  Assume that $\xi_1$ and $\xi_2$ satisfy~\ref{prop:4} with $\mathcal{K}_\infty$ functions $q_1$ and $q_2$ respectively. For all $\lambda\neq 0$ and all $Z\in\Re^{a\times b}$, \eqref{eq:homog_gen} yields 
\begin{equation}
\xi_i(Z)\geq \min_{j\in\{1,2\}} q_j\left(\dfrac{1}{|\lambda |}\right) \xi_i(\lambda Z). 
\end{equation} 
If we define $q$ so that for all $\lambda\in\Re_{\geq 0}$, $q(\lambda)=\min_{i\in\{1,2\}}q_i(\lambda)$, 
then $q$ is a $\mathcal{K}_{\infty}$ function (see Lemma~\ref{lem:kf} above) such that for all $\lambda\neq 0$ and $Z\in\Re^{a\times b}$,
\begin{equation}
\xi_1(Z)+\xi_2(Z)\geq q\left(\dfrac{1}{|\lambda |}\right) \left(\xi_1(\lambda Z)+\xi_2(\lambda Z)\right)
\end{equation}
therefore $\xi_1+\xi_2$ verifies property~\ref{prop:4}. Besides, for all $\lambda\neq 0$ and $Z$ in $\Re^{c\times d}$,
\begin{equation}
\xi_1(\ell(Z))\geq q_1\left(\dfrac{1}{|\lambda |}\right) \xi_1(\lambda \ell(Z))=q_1\left(\dfrac{1}{|\lambda |}\right) \xi_1(\ell(\lambda Z))
\end{equation}
given the linearity of $\ell$. We can then conclude that $\xi_1\circ\ell$ also verifies property~\ref{prop:4}.

\noindent \textit{Proof of (jj):}  Assume that $\xi_1$ and $\xi_2$ satisfy~\ref{prop:der} for $\gamma_1$ and $\gamma_2$  respectively. Let $\gamma=\min_{i\in\{1,2\}}\gamma_i $. Similarly to the first case, for all $Z_1,Z_2$ in $\Re^{a\times b}$ and $i$ in $\{1,2\}$,~\eqref{eq:tri_gen} yields 
\begin{equation}
\xi_i(Z_1-Z_2)\geq \gamma\xi_i(Z_1)-\xi_i(Z_2)
\end{equation}
which gives
\begin{equation}
\xi_1(Z_1-Z_2)+\xi_2(Z_1-Z_2)\geq \gamma \left(\xi_1(Z_1)+\xi_2(Z_1)\right)-\left(\xi_1(Z_2)+\xi_2(Z_2)\right)
\end{equation}
therefore $\xi_1+\xi_2$ satisfies property~\ref{prop:der}.  Moreover, for all $Z_1$ and $Z_2$ in $\Re^{c\times d}$,
\begin{equation}
\xi_1(\ell(Z_1-Z_2))=\xi_1(\ell(Z_1)-\ell(Z_2))\geq \gamma\xi_1(\ell(Z_1))-\xi_1(\ell(Z_2))
\end{equation}
so $\xi_1\circ \ell$ satisfies~\ref{prop:der} too.
{\flushright\qed}

\subsection{Proof of Lemma~\ref{lem:obs}}\label{app:obs_proof}
(i) $\Rightarrow$~(ii): Assuming that the system is observable on the interval $\interval{0}{T-1}$, we need to prove that there exists a $\mathcal{K}_\infty$ function $q$ which verifies~\eqref{eq:lem_lo_i}. The idea of the proof is to apply Lemma~\ref{lem:fc} to the function $F$ of $\Re^{n\times T}$ defined by $F(Z)=V_\Sigma(0,Z)$ with $V_\Sigma$ defined as in \eqref{eq:V}. 
To begin with, we note that $F$ can be decomposed as
$F=\xi\circ\ell$
where $\xi:\Re^{n\times (T-1)}\times\Re^{n\times T}\rightarrow\Re_{\geq 0}$ is a loss function such that for $Z=\begin{pmatrix}
z_{0} & \cdots & z_{T-2}
\end{pmatrix}$ in $\Re^{n\times (T-1)}$, $Y=\begin{pmatrix}
y_{0} & \cdots & y_{T-1}
\end{pmatrix}$ in $\Re^{n_y\times T}$,
$$\xi(Z,Y)=\sum_{t=0}^{T-2}\phi_t(z_{t})+\sum_{t=0}^{T-1}\psi_t(y_{t})$$
and $\ell:\Re^{n\times T}\rightarrow\Re^{n\times (T-1)}\times\Re^{n\times T}$ a linear mapping such that for all $Z=\begin{pmatrix}
z_{0}& \cdots & z_{T-1} 
\end{pmatrix}$ in $\Re^{n\times T}$,
\begin{multline*}
\ell(Z)=\Big(\begin{pmatrix}
z_{1}-A_{0}z_{0} & \cdots & z_{T-1}-A_{T-2}z_{T-2}
\end{pmatrix}\\
,\begin{pmatrix}
C_{0}z_{0} & \cdots & C_{T-1}z_{T-1}
\end{pmatrix}\Big). 
\end{multline*}
To apply Lemma~\ref{lem:fc} to $F$,  we need to check that $F$ fulfills the properties \ref{prop:prem}--\ref{prop:3}. In virtue of the assumptions on $\phi_t$ and $\psi_t$ agreed in the statement of the lemma, the first two properties are obviously satisfied. The third will be satisfied if $\ell$ is injective, a propriety which we now check. Let $Z$ be such that $\ell(Z)=0$. Then 
\begin{align}
&\forall t \in \{0 ,\ldots ,T-2\},&\: z_{t+1}-A_tz_t&=0 \label{eq:lem_obs_inj_i}\\
&\forall t \in \{0 ,\ldots ,T-1\}, &\: C_tz_t&=0 \label{eq:lem_obs_inj_ii}
\end{align}
An immediate consequence of \eqref{eq:lem_obs_inj_i}--\eqref{eq:lem_obs_inj_ii} is that $\mathcal{O}_{0,T-1}z_{0}=0$ which yields $z_{0}=0$ because the system is observable on $\interval{0}{T-1}$. Therefore, thanks to the recursive relation~\eqref{eq:lem_obs_inj_i}, we can conclude that $Z=0$, and so, the linear mapping $\ell$ is injective. \\
We can therefore apply Lemma~\ref{lem:fc} to conclude  that $F$ satisfy indeed ~\ref{prop:prem}--\ref{prop:4}. Now, consider a matrix norm $\lVert\cdot\rVert_{\text{ind}}$ on $\Re^{n\times T}$ induced by two vector norms $\lVert\cdot\rVert_T$ and $\lVert\cdot\rVert$ defined respectively on $\Re^T$ and $\Re^n$ in the sense that $$\lVert Z\rVert_{\text{ind}}=\sup_{\substack{\eta\in \Re^T\\\eta\neq 0}}\dfrac{\left\|Z\eta\right\|}{\left\|\eta\right\|_T}$$
Applying Lemma~\ref{lem:lb} to $F$ with the so-defined induced  norm, we infer that there  exists $D>0$ defined as in~\eqref{eq:def_d} and  a $\mathcal{K}_{\infty}$ function $q'$, such that for all $Z$ in $\Re^{n\times T}$,
\begin{equation}
F(Z)\geq Dq'(\lVert Z\rVert_{\text{ind}})
\end{equation}
If we denote with $e_1$ the canonical vector of $\Re^{T}$  with all entries equal to zero except the first one which is equal to $1$, then $Ze_1= z_{0}$. However, by definition of the induced norm, we know that $\lVert Ze_1\rVert/\lVert e_1\rVert_T\leq \lVert Z\rVert_{\text{ind}}$. Therefore, as $q'$ is an increasing  function, we get that $q'(\lVert z_{0}\rVert/\lVert e_1\rVert_T)\leq q'(\lVert Z\rVert_{\text{ind}})$. By posing $q:\lambda\mapsto Dq'(\lambda/\lVert e_1\rVert_T)$, it is easy to see that $q$ is a $\mathcal{K}_{\infty}$ function so that for all $Z$ in $\Re^{n\times T}$, $V(0,Z)=F(Z)\geq q(\lVert z_{0}\rVert)$. 
%
(ii) $\Rightarrow$~(i): Assume that there exists $q$ in $\mathcal{K}_{\infty}$ such that for all $Z~=~\begin{pmatrix}z_{0} & z_{1} &\ldots &z_{T-1}\end{pmatrix}$ in $\Re^{n\times T}$ such that \eqref{eq:lem_lo_i}  holds. 
We want to prove that the matrix $\mathcal{O}_{0,T-1}$ defined in \eqref{eq:Obsv-Matrix} is of full column rank, which is equivalent to showing that for $z$ in $\Re^n$, $\mathcal{O}_{0,T-1}z=0$ implies $z=0$.
For all $z\in\Re^n$, construct a sequence $Z^*=\begin{pmatrix}z_0^* & \cdots & z_{T-1}^*\end{pmatrix}$ as follows: $z_{0}^*=z$ and $z_{t+1}^*=A_tz_t^*$ for all $ t\in\{0,\ldots,T-2\}$.  
Since the inequality~\eqref{eq:lem_lo_i} is supposed to be true for any sequence, so it is for the particular sequence $\{z_t^*\}$ defined above. Applying this inequality to $Z^*$ yields
\begin{equation}\label{eq:lem_lo_p}
V(0,Z^*)=\sum_{t=0}^{T-1}\psi_t(C_tz_t^*)\geq q(\lVert z_{0}^*\rVert)
\end{equation}
Now, observe that if $\mathcal{O}_{0,T-1}z=0$, then it follows from the recursive relation $z_{t+1}^*=A_tz_t^*$ that for all $t$ in $\{0,\ldots,T-1\}$, $C_tz_t^*=0$. Injecting this  in~\eqref{eq:lem_lo_p} imposes that $q(\lVert z_{0}^*\rVert)\leq 0$ which necessarily implies that $z=0$ as $q$ is a $\mathcal{K}_{\infty}$ function. Therefore, the matrix $\mathcal{O}_{0,T-1}$ is injective and the system is observable on the interval $\interval{0}{T-1}$.

\subsection{Technical results for proving Corollary \ref{coro:Approximate-Resilience}}\label{subsec:appendix-Corollary}
\noindent This  section contains some technical steps of the proof of Corollary \ref{coro:Approximate-Resilience}. 
\begin{lem}\label{lem:exponential-cost-P1-P5}
If $\ell:\Re^{n_y}\rightarrow\Re_{\geq 0}$ satisfies \ref{prop:prem}--\ref{prop:3} and \ref{prop:der}, then  so does the function $\psi$ defined by 
$\psi(y)=1-e^{-\ell(y)}$.
Moreover if $\ell$ fulfills \ref{prop:4}, then $\psi$ satisfies the same property but with a function $q$ in $\mathcal{K}_{\sat,a}$ for $a=1$. 
\end{lem}
\begin{proof}
It is straightforward to check that $\psi$ obeys \ref{prop:prem}-\ref{prop:3}. By assumption, $\ell$ obeys \ref{prop:der}. Denote therefore  the associated constant with $\gamma_\ell$ (which, by \eqref{eq:tri_gen}, is necessarily less than or equal to $1$).  To see then that \ref{prop:der} is also satisfied by $\psi$, we just need to check that 
\begin{equation}\label{eq:psi(a+b)}
	\psi(a+b)-\bar{\gamma}_\ell \psi(a)-\bar{\gamma}_\ell\psi(b)\leq 0 \quad \forall (a,b)\in \Re^{n_y}\times \Re^{n_y}
\end{equation}
with $\bar{\gamma}_\ell=\gamma_\ell^{-1}\geq 1$, 
which is equivalent to 
$$1-2\bar{\gamma}_\ell+\bar{\gamma}_\ell e^{-\ell(a)}+\bar{\gamma}_\ell e^{-\ell(b)}-e^{-\ell(a+b)}\leq 0 $$
Noting that $\ell(a+b)\leq \bar{\gamma}_\ell\ell(a)+\bar{\gamma}_\ell\ell(b)$, we have $-e^{-\ell(a+b)}\leq -e^{-\bar{\gamma}_\ell\ell(a)-\bar{\gamma}_\ell\ell(b)}$. From this it follows that for  \eqref{eq:psi(a+b)} to hold, it is enough that
$$1-2\bar{\gamma}_\ell+\bar{\gamma}_\ell e^{-\ell(a)}+\bar{\gamma}_\ell e^{-\ell(b)}-e^{-\bar{\gamma}_\ell\ell(a)-\bar{\gamma}_\ell\ell(b)}\leq 0 $$ 
Posing $\alpha=e^{-\ell(a)}$ and $\beta=e^{-\ell(b)}$, it suffices that 
$$1-2\bar{\gamma}_\ell+\bar{\gamma}_\ell\alpha+\bar{\gamma}_\ell\beta-(\alpha\beta)^{\bar{\gamma}_\ell}\leq 0 \quad \forall (\alpha,\beta)\in \interval[open left]{0}{1} $$
which can indeed be checked to be true by applying the identity $1+\bar{\gamma}_{\ell}\alpha-\bar{\gamma}_{\ell}\leq \alpha^{\bar{\gamma}_{\ell}} $, see e.g., \cite[Fact 1.9.2]{Berstein09-Book}. In effect, it follows from this identity that
$$\begin{aligned}
	1-2\bar{\gamma}_\ell+\bar{\gamma}_\ell\alpha+&\bar{\gamma}_\ell\beta-(\alpha\beta)^{\bar{\gamma}_\ell}  \\
	&=(1+\bar{\gamma}_\ell\alpha-\bar{\gamma}_\ell)+(1+\bar{\gamma}_\ell\beta-\bar{\gamma}_\ell)-1-(\alpha\beta)^{\bar{\gamma}_\ell}\\
	&\leq \alpha^{\bar{\gamma}_\ell}+\beta^{\bar{\gamma}_\ell}-1-(\alpha\beta)^{\bar{\gamma}_\ell}\\
	& = -(1-\alpha^{\bar{\gamma}_\ell})(1-\beta^{\bar{\gamma}_\ell})\leq 0
\end{aligned}.$$
 In conclusion, \eqref{eq:psi(a+b)} holds and therefore $\psi$ satisfies \ref{prop:der}. \\
It remains now to check \ref{prop:4}. This follows directly from Lemma \ref{lem:Infimum} below, from which we know that 
 $\psi(y)\geq q^\star(1/\lambda) \psi(\lambda y)$ with  $q^\star$ is a saturated function in $\mathcal{K}_{\sat,1}$. 
\end{proof}

\begin{lem}\label{lem:Infimum}
Let $\ell:\Re^{n_y}\rightarrow\Re_{\geq 0}$ be a function  satisfying  properties \ref{prop:prem}--\ref{prop:2} and \ref{prop:4}. In particular, assume that property \ref{prop:4} is satisfied by $\ell$ with a  $\mathcal{K}_\infty$ function  $q$ such that \eqref{eq:homog_gen} is an equality relation.
Let $$g(y,\lambda)=\dfrac{1-e^{-\ell(y)}}{1-e^{-\ell(y/\lambda)}} $$ 
for $\lambda\neq 0$ and $y\neq 0$. 
Then the function $q^\star: \Re_{\geq 0}\rightarrow \interval{0}{1} $ defined  by
$q^\star(\lambda)=\inf_{y\neq 0} g(y,\lambda)$ for $\lambda>0$ and $q^\star(0)=0$,
is well-defined, continuous and strictly increasing on $\interval{0}{1}$. Moreover we have 
$$1-e^{-\ell(y)}\geq q^\star(1/\lambda) \big(1-e^{-\ell(\lambda y)}\big)\quad  \forall (\lambda,y)\in \Re_{>0}\times \Re^{n_y}$$
\end{lem}
\begin{proof}
Since $g$ is positive on its domain  (hence lower-bounded), the defining infimum of $q^\star$ is well-defined.  Pose $a=e^{-\ell(y)}$. Then by using the continuity property of $\ell$ and its radial unboundedness (see Lemma \ref{lem:lb}), we see that the range of $a$ when $y$ lives in $\Re^{n_y}\setminus\{0\}$ is $\interval[open]{0}{1}$.  From the assumptions of the lemma, $\ell(y/\lambda)=q(1/\lambda)\ell(y)$ for all $y$ and all $\lambda>0$  and so, $q(1)=1$ and $e^{-\ell(y/\lambda)}=a^{q(1/\lambda)}$.     
For all $\lambda> 0$ we can write
$$\begin{aligned}
	q^\star(\lambda)=\inf_{y\neq 0}g(y,\lambda)=\inf_{a\in \interval[open]{0}{1}}\dfrac{1-a}{1-a^{q(\frac{1}{\lambda})}}
\end{aligned}$$ 
with $q(1/\lambda)\geq 1$ for $0<\lambda\leq 1$ and $q(1/\lambda)< 1$ for $\lambda>1$. 
We therefore obtain 
$$q^\star(\lambda)=\left\{\begin{array}{ll}\dfrac{1}{q(1/\lambda)} & \mbox{if } 0<\lambda\leq 1 \\ 1 & \mbox{otherwise}\end{array}\right.$$
The so obtained $q^\star$ is clearly continuous wherever it is well defined. Moreover, since $\lim_{\lambda\rightarrow 0}q^\star(\lambda)=q^\star(0)=0$, we conclude that $q^\star$ is continuous on its entire domain. From the properties of $q$, we deduce that $q^\star$ is strictly increasing  on $\interval{0}{1}$.  Lastly, we observe that the inequality in the statement of the lemma is a direct consequence of the definition of $q^\star$. 
\end{proof}
\subsection{Proof of Lemma~\ref{lem:pr_calc}}\label{app:pr_proof}
The starting point of the proof is the observation that for every integer $r$ in $\left\{1,\ldots,T\right\}$, $\tilde{p}_r\leq r\tilde{p}_1$.  Hence it suffices to show that  $\tilde{p}_1=1/b_1$ and is as expressed in \eqref{eq:def_b1}. Recall that by definition,
\begin{equation}\label{eq:def_p1}
\tilde{p}_1=\sup_{(t,i )\in\mathbb{T}\times\mathbb{S}}\:\sup_{\substack{Z\in \Re^{n\times T}\\ Z \neq 0}}\dfrac{| c_{ti}^\top z_t|}{H_\Sigma(Z)}.
\end{equation}
Without loss of generality, assume that $c_{ti}^\top \neq 0$ for all $(t,i)\in \mathbb{T}\times \mathbb{S}$ 
Then for any $(t,i)$,  
$$\begin{aligned}
	\sup_{\substack{Z\in \Re^{n\times T}\\ Z \neq 0}}\dfrac{| c_{ti}^\top z_t|}{H_\Sigma(Z)}
	& =\sup_{\substack{Z\in \Re^{n\times T}\\ Z \neq 0}}\left\{\dfrac{|c_{ti}^\top z_t|}{H_\Sigma(Z)}:c_{ti}^\top z_t\neq 0\right\}
	\triangleq \dfrac{1}{\beta_{ti}}
\end{aligned}$$
where 
$$\begin{aligned}
	\beta_{ti}&=\inf_{\substack{Z\in \Re^{n\times T}\\ Z \neq 0}}\left\{\dfrac{H_\Sigma(Z)}{| c_{ti}^\top z_t|}:c_{ti}^\top z_t\neq 0\right\} \\
	&=\inf_{\substack{Z\in \Re^{n\times T}\\ Z \neq 0}}\left\{H_\Sigma(Z):|c_{ti}^\top z_t|=1\right\} \\
	&=\inf_{\substack{Z\in \Re^{n\times T}\\ Z \neq 0}}\left\{H_\Sigma(Z):c_{ti}^\top z_t=1\right\} 
\end{aligned}$$
Recalling that $H_\Sigma(Z)$ is a norm under the conditions of the lemma,  the second equality  in the expression of $\beta_{ti}$ above follows from the (strict) homogeneity property of norms. As to the last equality, it follows from the fact that $c_{ti}^\top z_t$ is a scalar which induces the possibility to replace the constraint $|c_{ti}^\top z_t|=1$ indifferently either by $c_{ti}^\top z_t=1$  or by $c_{ti}^\top z_t=-1$. \\
Now by invoking the definition of $\tilde{p}_1$, it can be seen that 
$$\tilde{p}_1=\sup_{(t,i )\in\mathbb{T}\times\mathbb{S}}\dfrac{1}{\beta_{ti}}=\dfrac{1}{\inf_{(t,i )\in\mathbb{T}\times\mathbb{S}}\beta_{ti}}=\dfrac{1}{b_1}.$$
\qed

\balance
\bibliographystyle{abbrv}

\end{document}